\documentclass[10pt,twocolumn,twoside]{IEEEtran} 
\usepackage{setspace}

\usepackage[T1]{fontenc}    % Vectorial codification T1
\usepackage[utf8]{inputenc} % on veut nos accents !

\usepackage{graphicx} 
\graphicspath{{./figures/}}
\DeclareGraphicsExtensions{.pdf,.png,.eps,.ps}
\usepackage{caption}
\usepackage{subcaption}
\usepackage{amsmath,amssymb} % amssymb internally loads amsfonts
\usepackage{bbm}
\usepackage{enumerate}
\usepackage{float}
\usepackage{color}
\usepackage[Algorithm]{algorithm}
\usepackage{algpseudocode}
\algdef{SE}[DOWHILE]{Do}{doWhile}{\algorithmicdo}[1]{\algorithmicwhile\ #1}%

\usepackage{amsthm}

\usepackage[hidelinks]{hyperref}

\theoremstyle{break}
\newtheorem{theorem}{Theorem}

\newcommand{\node}{N}
\newcommand{\link}{L}
\newcommand{\linkset}{\mathcal{L}}
\newcommand{\junction}{J}
\newcommand{\phases}{\mathcal{P}}
\newcommand{\nodeset}{\mathcal{N}}
\newcommand{\junctionset}{\mathcal{J}}
\newcommand{\NN}{\mathbb{N}}
\newcommand{\inputnodes}{\mathcal{I}}
\newcommand{\outputnodes}{\mathcal{O}}

\newcommand{\pbf}{\mathbf{p}}

\title{Capacity-aware back-pressure\\ traffic signal control}

\author{Jean Gregoire \hspace{3mm} Xiangjun Qian \hspace{3mm} Emilio Frazzoli \hspace{3mm} Arnaud de La Fortelle \hspace{3mm} Tichakorn Wongpiromsarn}

\begin{document}
%\doublespacing
\maketitle

%%%%%%%%%%%%%%%%%%%%%%%%%%%%%%%%%%%%%%%%%%%%%%%%%%%%%%%%%%%%%%%%%%%%%%%%%%%%%%%%
\begin{abstract}
The control of a network of signalized intersections is considered. Previous work demonstrates that the so-called back-pressure control provides stability guarantees, assuming infinite queues capacities. In this paper, we highlight the failing of current back-pressure control under finite capacities by identifying sources of non work-conservation and congestion propagation. We propose the use of a normalized pressure which guarantees work conservation and mitigates congestion propagation, while ensuring fairness at low traffic densities, and recovering original back-pressure as capacities grow to infinity. This capacity-aware back-pressure control allows to improve performance as congestion increases, as indicated by simulation results, and keeps the key benefits of back-pressure: ability to be distributed over intersections and $\mathcal{O}(1)$ complexity.
\end{abstract}

\section{Introduction}

Congestion is one of the major problems in today's metropolitan transportation networks. Before investigating investments in order to enhance the capacity of the network, or policies to reduce the traffic load, one must wonder whether the network is used at its maximum capacity. Vehicle automation is expected to enable much more precise and intelligent coordination between vehicles, possibly reducing congestion~\cite{Dresner2008}. However, automated cars are not currently ready for large commercial deployment. Human-driven cars can only be coordinated by traffic signals: more complex scheduling at intersections would require automation to be safe. That is why it is of high interest to study the theoretical maximum throughput of a network of intersections coordinated by traffic lights.

Traffic lights at intersections alternate the right-of-way of users (e.g., cars, public transport, pedestrians) to coordinate conflicting flows. A particular set of feasible simultaneous rights of way, called a phase, is decided for a certain period of time~\cite{Papageorgiou2003}. Controlling a traffic light consists of designing rules to decide which phase to apply over time. 

Pre-timed policies activate phases according to a time-periodic pre-defined schedule. There is much previous work on designing optimal pre-timed policies. However, such policies are not efficient under changing arrival rates which require adaptive control. Most used adaptive traffic signal control systems include SCOOT~\cite{Hunt1982}, SCATS~\cite{Lowrie1990}, PRODYN~\cite{Henry1984}, RHODES~\cite{Mirchandani2001}, OPAC~\cite{Gartner1983} or TUC~\cite{Diakaki2002}. These systems update some control variables of a configurable pre-timed policy on middle term, based on traffic measures, and apply it on short term. Control variables may include phases, splits, cycle times and offsets~\cite{Papageorgiou2003}. Such algorithms may differ in the way optimization is carried out (e.g., linear/dynamic programming, exhaustive enumeration) and in the modeling approach (e.g., queuing network model~\cite{Osorio2009}, cell transmission model~\cite{Lo2001}, store-and-forward~\cite{Aboudolas2009}, petri nets~\cite{DiFebbraro2002}).  Many major cities currently employ these systems which proved to be able to yield various benefits, including travel time and fuel consumption reduction, as well as safety improvements~\cite{Shepherd1992}.

More recently, based on the seminal paper~\cite{Tassiulas1992}, feedback controls have been proposed both in the case of deterministic arrivals~\cite{Varaiya2013}, or stochastic arrivals~\cite{Varaiya2009, Wongpiromsarn2012}. Time is slotted and at every time slot, a feedback controller decides the phase to apply based on current queue length estimation. This requires real-time queues estimation, but it enables to be much more reactive than other traffic controllers and to have stability guarantees. Reference~\cite{Tassiulas1992}  introduced the so-called back-pressure control which computes the control to apply based on queue lengths, and can achieve provably maximum stability. This algorithm was originally applied to wireless communication networks~\cite{Neely2003,Neely2005}, and some effort has been required to apply the approach in the context of a network of intersections~\cite{Varaiya2009, Wongpiromsarn2012}. A key feature of this algorithm is that it can be completely distributed over intersections, in the sense that it can be implemented by running an algorithm of complexity $\mathcal{O}(1)$, requiring only local information, at each intersection. 

However, the strong assumption of current back-pressure traffic control algorithms is the unboundedness of queue capacities. Indeed, when the queue at the entry of an intersection grows so much that it reaches the upstream intersection, congestion will propagate: this is a non-negligible and easy to observe phenomenon. The phenomenon is commonly referred as blocking in queuing theory, and many blocking types can be considered~\cite{Perros1994}. In worst-case scenario, blocking results in deadlocks whose resolution can be of high complexity~\cite{Kundu1989, gregoire:hal-00828976}. An off-line optimization of a pre-timed policy is proposed in~\cite{Osorio2009,Osorio2009a}. The standard queuing network model with fixed service times of servers~\cite{Bocharov2003} is modified to account for blocking causing inter-queue interactions. The notion of effective service rate aims at accounting for both service and blocking. An expression of the blocking probability of each queue, i.e., the probability of the queue to be full, can be derived. The idea is to include in the optimized objective function penalties for high blocking probability. The method proved to be efficient to improve performance as congestion increases. However, it is for off-line optimization of fixed-cycle signals for a certain scenario (given arrival rates). In this paper, we aim at building a feedback control that can adapt on-line to varying situations. Some works applied to wireless communication networks have proposed feedback controls that can achieve maximum throughput under queue boundedness constraints~\cite{Giaccone2007}. However, they suppose the absence of arrivals at internal nodes, cannot be easily implemented, and are thus not suitable for our application. 

This paper proposes to keep the fundamental idea of back-pressure control, that is pressure computation at every node of the network, in order to keep the resulting key benefits: ability to be distributed over intersections and $\mathcal{O}(1)$ complexity. However, we propose to take into account the queue capacities for the computation of pressures. The idea is to normalize pressures, so that full queues all exert the same normalized maximal pressure independently from their capacity. Following the idea of~\cite{Giaccone2007}, this normalization is expected to decrease the blocking probability.

The paper is organized as follows. Section~\ref{sec:model} describes the phase-based queuing network model. Section~\ref{sec:failing-bp} presents the current back-pressure traffic signal control of~\cite{Wongpiromsarn2012} and proves its lack of work conservation and its inability to avoid congestion propagation under finite queue capacities. Section~\ref{sec:capacity-aware-traffic-control} proposes the use of normalized pressures and proves the benefits in terms of work-conservation and congestion mitigation. Simulations of Section~\ref{sec:simulations} show the efficiency of the approach proposed in this paper and Section~\ref{sec:conclusions} concludes and opens perspectives.

\section{Model}
\label{sec:model}

\subsection{Queuing network topology}

The network of intersections is modelled as a directed graph of nodes $(\node_a)_{a\in\nodeset}$ and links $(\link_j)_{j\in\linkset}$. Nodes represent roads with queuing vehicles, and links enable transfers from node to node. This is a standard queuing network model.

\begin{figure}[ht]
\centering
\includegraphics[width=0.8\linewidth]{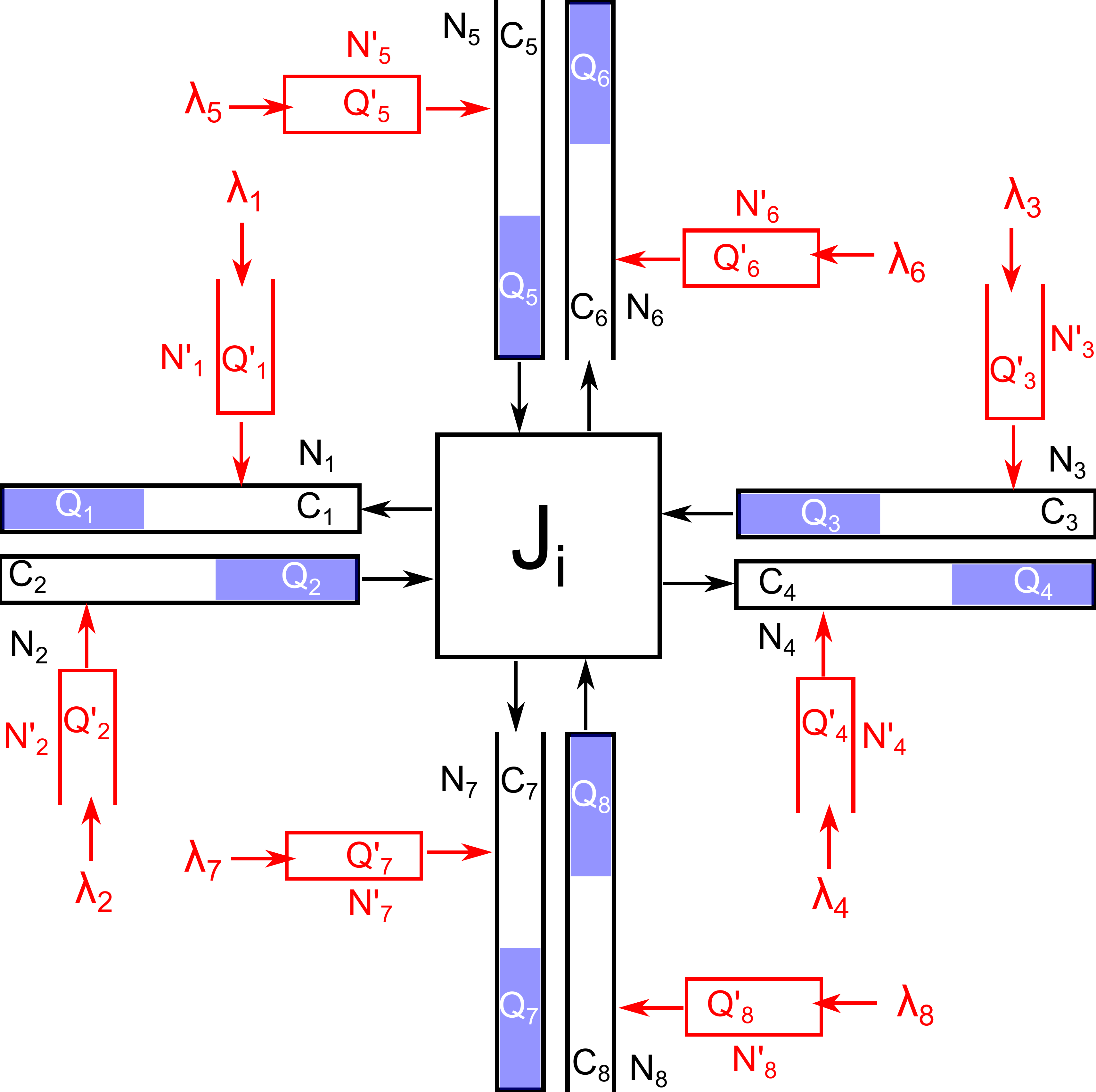}\hfill
\caption{A junction with 4 incoming nodes and 4 outgoing nodes which corresponds to the intersection depicted in Figure~\ref{fig:phases}.}
\label{fig:network-topology}
\end{figure}

Every signalized intersection is modelled as a server managing a junction which consists of set of links. Junctions $(\junction_i)_{i\in\junctionset}$ are supposed to form a partition of links. For every junction $\junction$, $\inputnodes(\junction)$ and $\outputnodes(\junction)$ denote respectively the inputs and the outputs of $\junction$. Inputs (resp.outputs) of junction $\junction$ are nodes $N$ such that there exists a link $L\in\junction$ pointing from (resp. to) $N$. The reader should consider the introduction of junctions in the model as an overlay of the queuing network model.

For the sake of simplicity, we do not represent links in the queuing network representation of Figure~\ref{fig:network-topology} and we assume that at every junction, there exists a link from any input to any output. If the link does not exist physically, the flow through it will be always zero.

Every server maintains an internal queue for every input/output, and server work enables to transfer vehicles from an input to an output of the junction. As standard in queuing network control, time is slotted, and each (time) slot $k\in\NN$ maps to a certain period of time. Due to routing of vehicles, there are several queues at node $\node_a$ and $Q_{ab}(k)$ denotes the number of vehicles at node $\node_a$ in slot $k$ waiting to leave node $\node_a$ for $\node_b$. $Q_a(k):=\sum_b Q_{ab}(k)$ denotes the total number of vehicles waiting at node $\node_a$.

\subsection{Arrival and routing processes}

Let $A_a(k)$ denote the number of vehicles that exogenously arrive at node $\node_a$ during slot $k$. We assume that the arrival process $A_a(k)$ is rate-convergent with rate $\lambda_a$ which represents the expected number of arrivals per time slot at node $N_a$  in the long-term. The arrival process is not controlled, it is an exogenous process. There are also endogenous vehicle transfers allowed by links at junctions. We let $f_{ca}(k)$ denote the number of vehicles leaving $N_c$ for $N_a$ during slot $k$. When a vehicle enters node $N_a$ at slot $k$ endogenously (originating from another node), or exogenously (inserted into the network through the arrival process), it increments one of the queues $Q_{ab}(k)$, unless it leaves the network at $N_a$. We assume that the ratio of vehicles added to $Q_{ab}(k)$ is rate-convergent with rate $r_{ab}\in[0,1]$. Rates $r_{ab}$ represent the long-term routing ratios of vehicles entering $N_a$. Because some vehicles leave the network at $N_a$ and are not added to some queue, the routing ratios do not necessarily sum to $1$ and $1-\sum_b r_{ab} \geq 0$ represents the exit rate at node $N_a$. The queue dynamics is as follows:
\begin{equation}
Q_{ab}(k+1)=Q_{ab}(k)-f_{ab}(k)+r_{ab}(k)\left(\sum_c f_{ca}(k) + A_a(k)\right)\text{,}
\end{equation}
where $r_{ab}(k)$ is rate-convergent with rate $r_{ab}$.

\subsection{Phase-based control}

At every time slot, the service offered by servers at junctions is controlled by activating a given signal phase $p_i$ at every junction $\junction_i$ from a predefined finite set of feasible phases $\phases_i$. When phase $p_i$ is activated during one slot, $\mu_{ab}(p_i)$ represents the maximum number of vehicles transferred from $\node_a$ to $\node_b$ during that slot. Figure~\ref{fig:phases} depicts the 4 typical phases of a 4 inputs/4 outputs junction. Each global phase $p=(p_i)_{i\in\junctionset}$ results in a different service $\mu(p)$.

\begin{figure}[ht]
\centering
\includegraphics[width=0.8\linewidth]{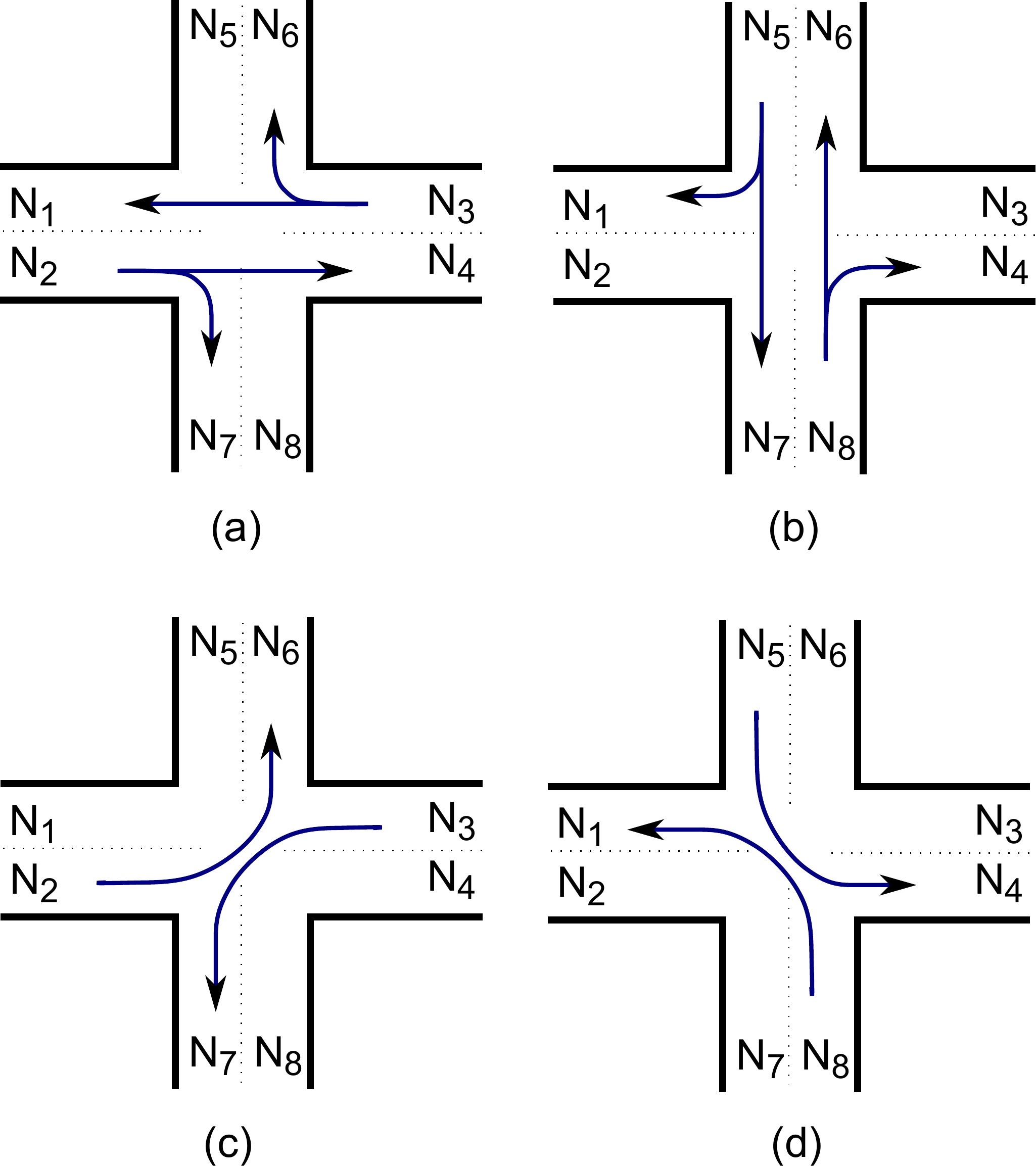}\hfill
\caption{A typical set of feasible phases at a junction. For phase (c), we have $\mu_{26}(p^{(c)})>0$ and $\mu_{37}(p^{(c)})>0$.}
\label{fig:phases}
\end{figure}

Two phenomena may affect the actual number of vehicles being transferred. First, only the vehicles which are currently at a queue at the beginning of the time slot can leave the queue during that slot. Second, and this is the phenomenon highlighted in this paper, above a certain queue length, a node is full and cannot accept vehicles any more.

Let $\pbf(k)$ denote the phase control through time (the only controlled variable), the number of vehicles transferred from $N_a$ to $N_b$ during slot $k$ is:
\begin{eqnarray}
f_{ab}(k)&=&\delta(Q_b(k),C_b) \min(Q_{ab}(k),\mu_{ab}(\pbf(k)) \label{eq:flow-with-blocking}
\end{eqnarray}

The function $\delta(q,c)$ models blocking due to downstream congestion, and  returns $1$ if $q < c$, and $0$ else. $C_b$ is referred as the capacity of node $N_b$: it is the maximum queue length from which the node cannot accept vehicles any more. We say that node $N_b$ is full at time slot $k$ if $Q_b(k) \geq C_b$. The blocking phenomenon is illustrated in Figure~\ref{fig:blocking}. Note that this simple binary model is conservative because in reality, even if a node is full at the beginning of the time slot, some vehicles may be able to enter this node if the downstream junction gives the right-of-way to vehicles in that node. This effect will not be considered in this paper. %However, due to inertia, the number of vehicles transferred to that node will be significantly affected anyway. 

\begin{figure}[ht]
\centering
\includegraphics[width=1\linewidth]{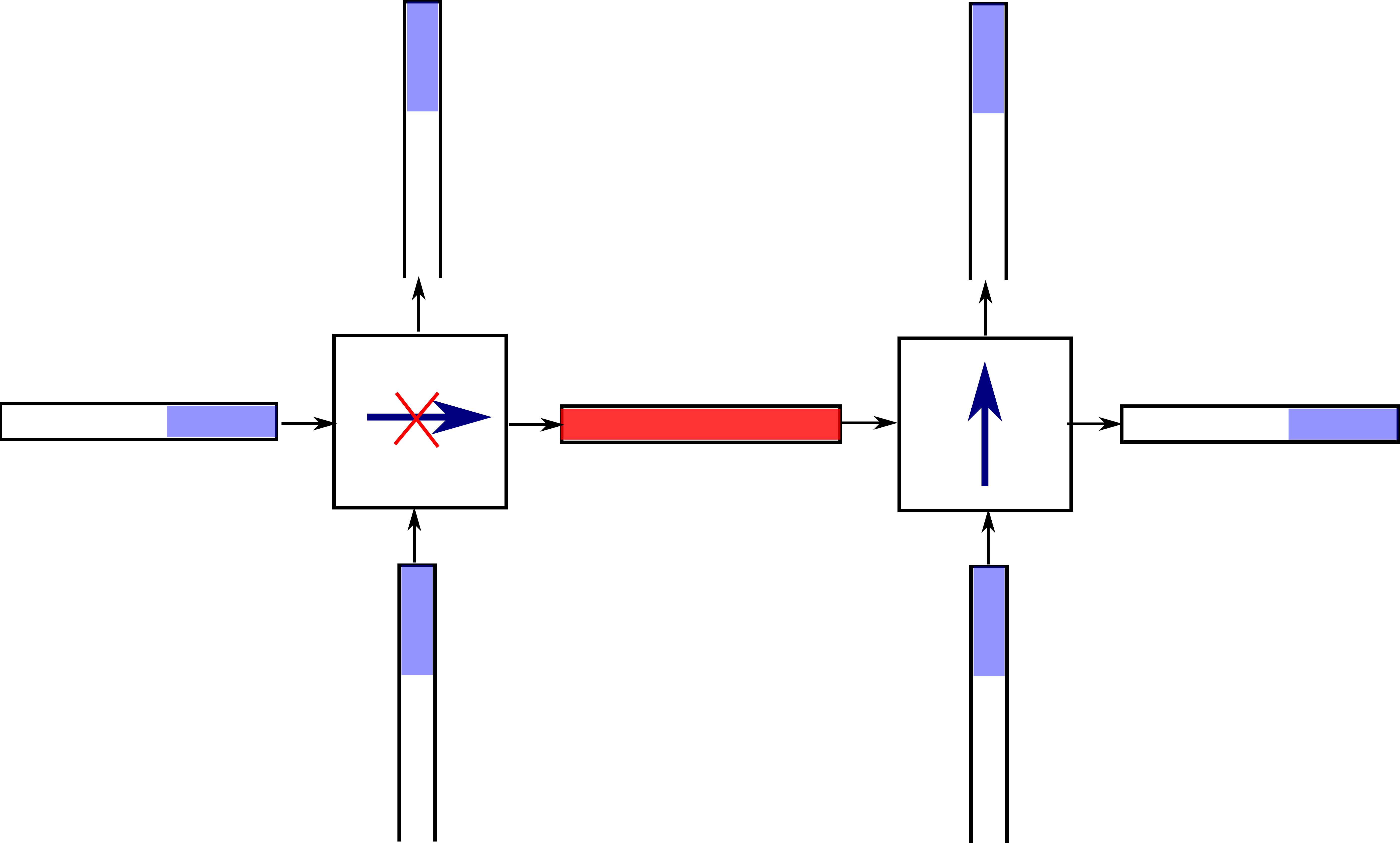}\hfill
\caption{Since the red-colored middle node is full, it cannot accept vehicles any more, even if the traffic light gives the right-of-way to vehicles going to this node.}
\label{fig:blocking}
\end{figure}

\section{Failing of back-pressure control under bounded queues constraints}
\label{sec:failing-bp}

\subsection{Back-pressure control}

Back-pressure control consists of a feedback control law that decides the phase to apply at every time slot based on current queues lengths. Let $\phi$ denote the control law, the phase control $\pbf$ satisfies at every time slot $k$:
\begin{equation}
\pbf(k)=\phi(Q(k))
\end{equation}
where $Q(k)$ denotes the network state at slot $k$, that is the queues lengths matrix $[Q_{ab}(k)]_{a,b\in\nodeset}$. The control law maps the current network state to the phase to apply. The back-pressure control law is defined component-wise for each junction as described by Algorithm~\ref{alg:back-pressure}. It only requires aggregated queues lengths $Q_a(k)$ and binary variables $d_{ab}(k)\in\{0,1\}$ indicating the presence of vehicles waiting at $N_a$ to leave $N_a$ for $N_b$. The latter can be measured using loop detectors on the dedicated lanes at the entry of the junction.
\begin{algorithm}[ht]
\begin{algorithmic}[5]
\For{$N_a\in \inputnodes(\junction_i) \cup \outputnodes(\junction_i)$}
\State $\Pi_{a}(k) \gets P_{a}\left(Q_{a}(k)\right)$
\EndFor
\For{$N_a\in \inputnodes(\junction_i), N_b\in\outputnodes(\junction_i)$}
\State $W_{ab}(k) \gets d_{ab}(k) \max\left(\Pi_{a}(k)-\Pi_{b}(k),0\right)$
\EndFor
\State $\pbf_i(k) \gets \arg \max\limits_{p_i\in\phases_i} \sum\limits_{a, b} W_{ab}(k)  \mu_{ab}(p_i)$ \label{line-arg-max-bp}
\State \Return{Phase $\pbf_i(k)$ to apply in time slot $k$ at junction $\junction_i$}
\end{algorithmic}
\caption{Back-pressure control law at junction $\junction_i$}
\label{alg:back-pressure}
\end{algorithm}

The idea of back-pressure control is to compute pressures at every node of the network based on queue lengths and to allow flows with a high upstream pressure and a low downstream pressure, like opening a tap. Current back-pressure traffic signal control~\cite{Wongpiromsarn2012, Varaiya2009} uses Algorithm~\ref{alg:back-pressure} with  linear pressure functions $P_a(Q_a)=Q_a$: the pressure exerted by a node equals its queue length. Algorithm~\ref{alg:back-pressure} proceeds as follows:
\begin{itemize}
\item First, pressures at input/output nodes are computed. If linear pressures are used, the pressure exerted by node $N_a$ at slot $k$ is $P_a(Q_a(k))=Q_a(k)$
\item Then, the pressure difference associated to each flow from an input $N_a$ to an output $N_b$ of the junction is computed, it equals $\max\left(\Pi_{a}(k)-\Pi_{b}(k),0\right)$, and it is multiplied by $d_{ab}(k)$ to account for the presence/absence of vehicles waiting at $N_a$ willing to leave $N_a$ for $N_b$.
\item For each phase $p_i\in\phases_i$, the total pressure release allowed by $p_i$ is computed: it equals the sum of pressure differences through each link of the junction weighted by the flow of vehicles that can be transferred through the corresponding link when phase $p_i$ is activated, that is $\sum\limits_{a, b} W_{ab}(k)  \mu_{ab}(p_i)$
\item Finally, the returned phase $\pbf_i(k)$ is the phase $p_i$ maximizing the weighted sum $\sum\limits_{a, b} W_{ab}(k)  \mu_{ab}(p_i)$.
\end{itemize}
This control law is proved to be stability-optimal under infinite capacities, i.e., the queuing network is stabilized for all arrival rates that can be stably handled considering all control policies. The two key properties of back-pressure control are its ability to be distributed over junctions and its $\mathcal{O}(1)$ complexity. 

Under bounded queues constraints, with linear pressure functions, pressure at $\node_a$ saturates at $P_a=C_a$ for $Q_a=C_a$. In the following, we show that this saturation at different levels for every node may result in a loss of work conservation and congestion propagation as presented in the sequel.

\subsection{Loss of work conservation}
First of all, the notion of work and work-conservation is defined in our context. We say that the server at junction $\junction_i$ works during slot $k$ if there are transfers through links of the junction during the slot. A control is work-conserving if the existence of an input $N_a$ and an output $N_b$ such that $Q_{ab}(k) > 0$ and $Q_b(k) < C_b$ is sufficient to ensure that the server of the junction works during slot $k$.   A loss of work-conservation is clearly a sign of inefficiency. Due to limited queues capacities, back-pressure control under linear pressure functions is not work-conserving as stated in Theorem~\ref{thm:non-work-conserving}, with a concrete example depicted in Figure~\ref{fig:non-work-conservation}.

\begin{theorem}[Loss of work-conservation under back-pressure]
\label{thm:non-work-conserving}
Under bounded queues constraints, back-pressure control is not work-conserving in the general case.
\end{theorem}
\begin{proof}
Consider the network of Figure~\ref{fig:non-work-conservation}. Suppose that the middle junction has two feasible phases: $p_{ab}$ with $\mu_{ab}(p_{ab})>0$, and $p_{cd}$ with $\mu_{cd}(p_{cd})>0$. Suppose that at time slot $k$, $Q_b(k)=C_b$, $Q_a(k)>C_b$ and $Q_c(k) < Q_d(k)$. Then, $W_{ab}(k) \mu_{ab}(p_{ab})>0$ and $W_{cd} \mu_{cd}(p_{cd}) = 0$. Hence, the phase to apply at the middle junction computed by Algorithm~\ref{alg:back-pressure} is $p_{ab}$. Since, $Q_b(k)=C_b$, Equation \eqref{eq:flow-with-blocking} implies that $f_{ab}(k)=0$. As a result, the middle junction will not work, because the selected phase is $p_{ab}$, but due to downstream congestion, transfers to $N_b$ are not feasible. However, choosing phase $p_{cd}$ would have enabled the server to work.
\end{proof}

\begin{figure}[ht]
\centering
\includegraphics[width=1\linewidth]{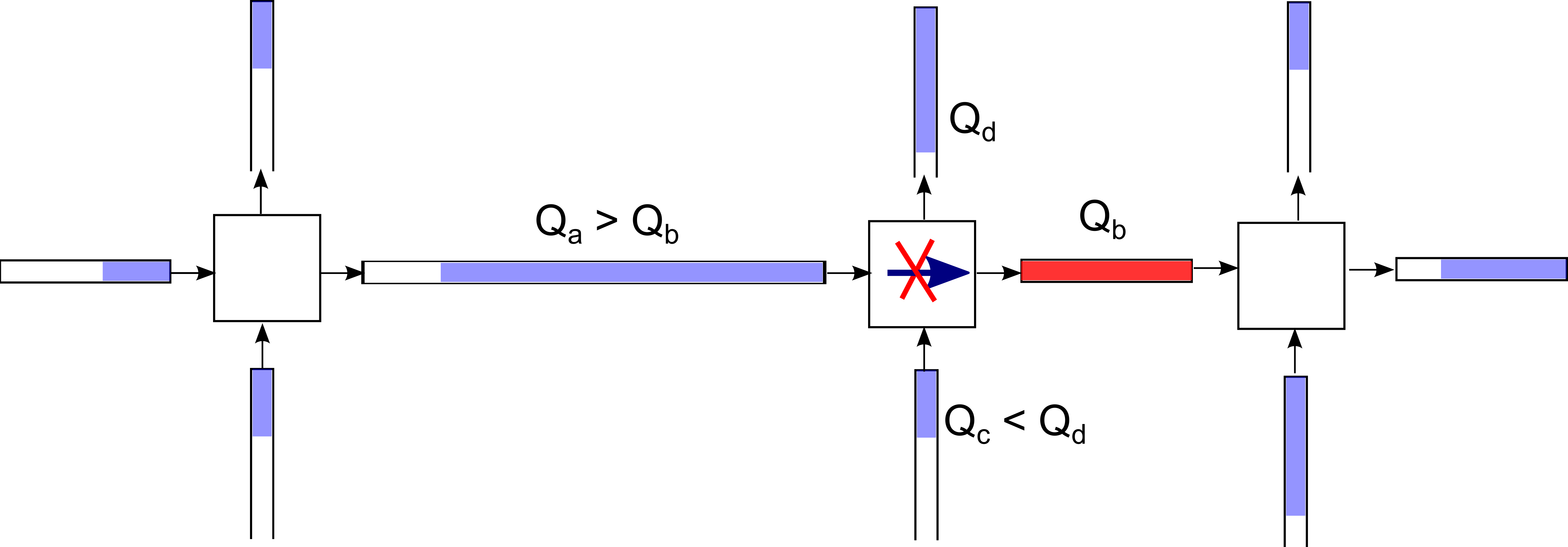}\hfill
\caption{Example of loss of work conservation of back-pressure control with linear pressure functions.}
\label{fig:non-work-conservation}
\end{figure}

The proved loss of work conservation is an important property since the subsequent inefficiency results in congestion propagation as highlighted in the following.

\subsection{Congestion propagation and deadlocks}

As depicted in Figure~\ref{fig:congestion-propagation}, loss of work conservation may result in congestion propagation, both to the node which has the right-of-way but cannot empty because of downstream congestion, and to the node which has not the right-of-way.

\begin{figure}[ht]
\centering
\includegraphics[width=1\linewidth]{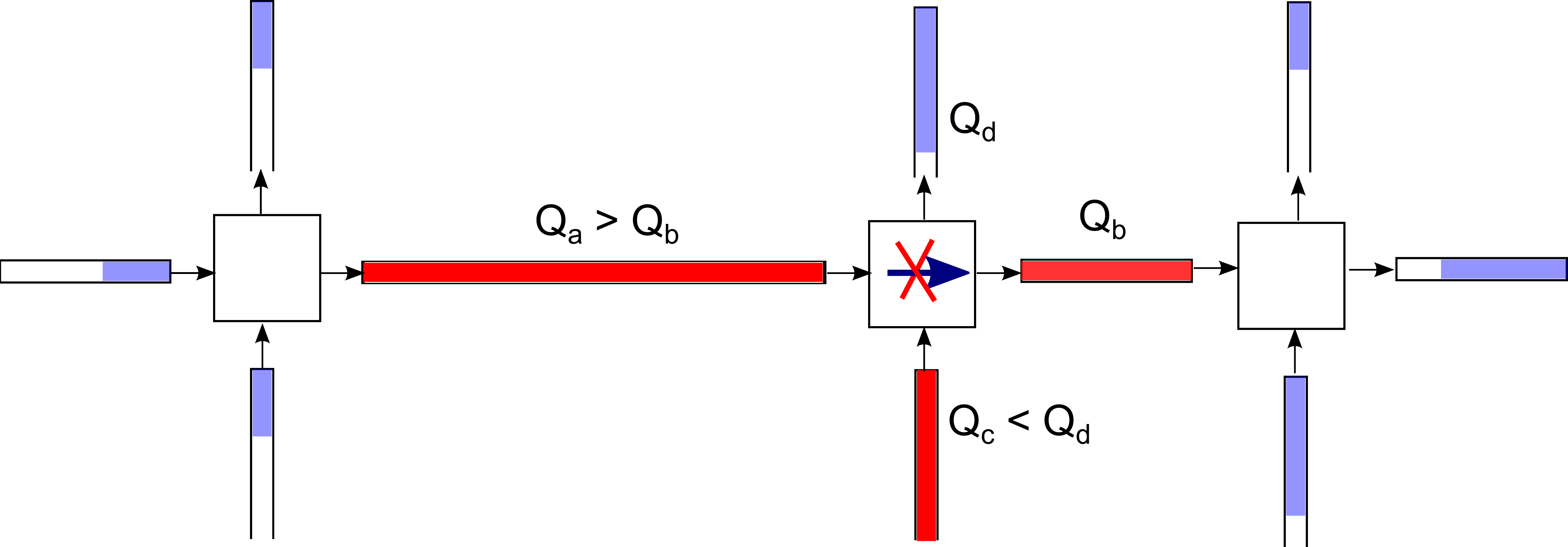}\hfill
\caption{Congestion propagation due to a loss of work conservation. Since $\node_b$ is full $f_{ab}(k)=0$ and $\node_a$ is not emptied, so congestion propagates to $\node_a$. Moreover, $\node_c$ is also not emptied because vehicles do not have the right-of-way under the selected phase, so congestion propagates to both nodes $\node_a$ and $\node_c$.}
\label{fig:congestion-propagation}
\end{figure}

In worst-case scenario, such congestion propagation can lead to deadlocks, as depicted in Figure~\ref{fig:deadlock}. 

\begin{figure}[ht]
\centering
\includegraphics[width=1\linewidth]{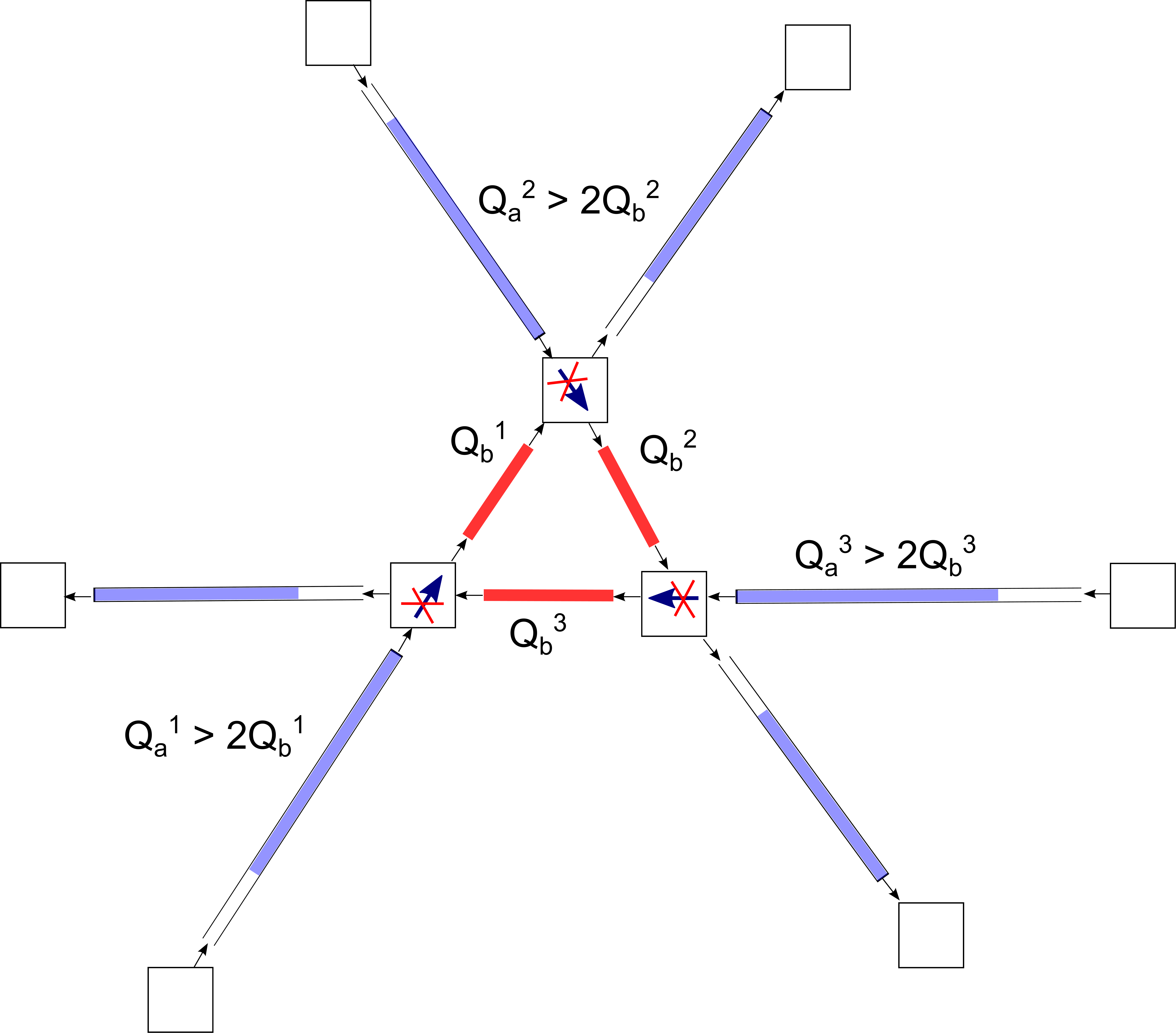}\hfill
\caption{A deadlock for back-pressure control. $Q_a^1$, $Q_a^2$ and $Q_a^3$ will be ever growing.}
\label{fig:deadlock}
\end{figure}

\section{Capacity-aware traffic control}
\label{sec:capacity-aware-traffic-control}

The following presents our approach to take into account the limited queues capacities by using normalized pressures in back-pressure control to enforce work-conservation and mitigate congestion propagation. First of all, the reasons that motivate devising convex normalized pressures are introduced in the sequel.

\subsection{Purpose of a convex normalized pressure and criteria}

\subsubsection{Purpose of a convex pressure}

When a node approaches full occupancy, every additional vehicle is more and more problematic as the queue grows. That is why it makes sense to define a pressure such that the marginal pressure, i.e., the increase in pressure due to an additional vehicle, rises as the queue grows. This remark justifies the use of a convex pressure.

\subsubsection{Purpose of a normalized pressure}

When a node $\node_b$ is full, it does not make sense to have a node $\node_a$ such that $P_a-P_b>0$, since the "tap" associated to the flow from $\node_a$ to $\node_b$ cannot be opened. That is why we propose to normalize the pressures so that any full node will exert a pressure that equals $1$, while an empty node does not exert any pressure. 

The simplest normalization that could be carried out would consist of using relative pressure functions ${P_a(Q_a)=Q_a/C_a}$. 
However, in order to have strictly convex pressures and to respect the fairness requirement proposed below, the normalization will be slightly more complex.

\subsubsection{Requirement for fairness at low traffic density}

At low traffic density, there is no reason to be unfair. Indeed, even a random choice of phases would stabilize the network. So, unfairness would not be justified by any global stabilization goal. Suppose that we use relative pressures functions ${P_a(Q_a)=Q_a/C_a}$, then an additional vehicle causes an increase in pressure of $1/C_a$, i.e., as high as capacity decreases. However, at low traffic density the marginal pressure should be uniform over nodes. We say that pressure functions $\{P_a(Q_a) : a\in\nodeset\}$ are fair at low traffic density if:
\begin{equation}
\exists K>0: \forall a\in\nodeset, \frac{dP_a}{dQ_a}(0)=K
\end{equation}

\subsubsection{Requirements for stability guarantees conservation}

Finally, it is important to ensure that as capacities grow to infinity, the original back-pressure control is recovered, to take advantage of the stability guarantees in the context of infinite capacities. That is why a requirement for the pressure function $P_a(Q_a)$ is to be linear for $Q_a/C_a \to 0$. If pressure functions are fair, $P_a(Q_a)=K Q_a + o_a(Q_a/C_a)$ in Landau notation, and the pressure function is linear at low occupancy. As a result, if the queues are always much under maximum occupancy, i.e., if the infinite queues capacities assumption is valid, the back-pressure control and its stability guarantees are recovered.

Now we have presented criteria that should respect the modified pressure in order to be capacity-aware and fair at low traffic density, we propose in the following a convex normalized pressure which respects the above criteria. 

\subsection{Example of normalized pressure}

The proposed pressure function should just be considered as an example of a pressure function fulfilling the presented requirements:

\begin{equation}
P_a(Q_a)=\min\left(1,~\frac{\frac{Q_a}{C_\infty}+\left(2-\frac{C_a}{C_\infty}\right)\left(\frac{Q_a}{C_a}\right)^m}{1+\left(\frac{Q_a}{C_a}\right)^{m-1}}\right)
\label{eq:pressure-function}
\end{equation}

At low occupancy, the pressure at node $N_a$ is linear: $P_a(Q_a) \simeq Q_a/C_\infty$, so pressure functions are fair and respect the requirement for stability guarantees conservation. The function is convex: the slope of the pressure increases as occupancy grows. Pressure over congestion threshold is normalized: $\forall a\in\nodeset, \forall Q_a\geq C_a, P_a(Q_a)=1$. The shape of pressure functions for two different capacities is depicted in Figure~\ref{fig:convex-normalized-pressure}. One can observe that the pressure function leaves the initial linear behavior at lower occupancy as capacity decreases.

The parameters $m$ and $C_\infty$ determine the shape of pressure functions; $m$ locates  the transition from the linear regime, while $C_\infty$ determines the slope of the pressure at low occupancy, and is such that a node which capacity is $C_\infty$ will have a linear pressure. We assume that all capacities are lower than $C_\infty$ and $m>1$.

\begin{figure}[ht]
\centering
\includegraphics[width=1\linewidth]{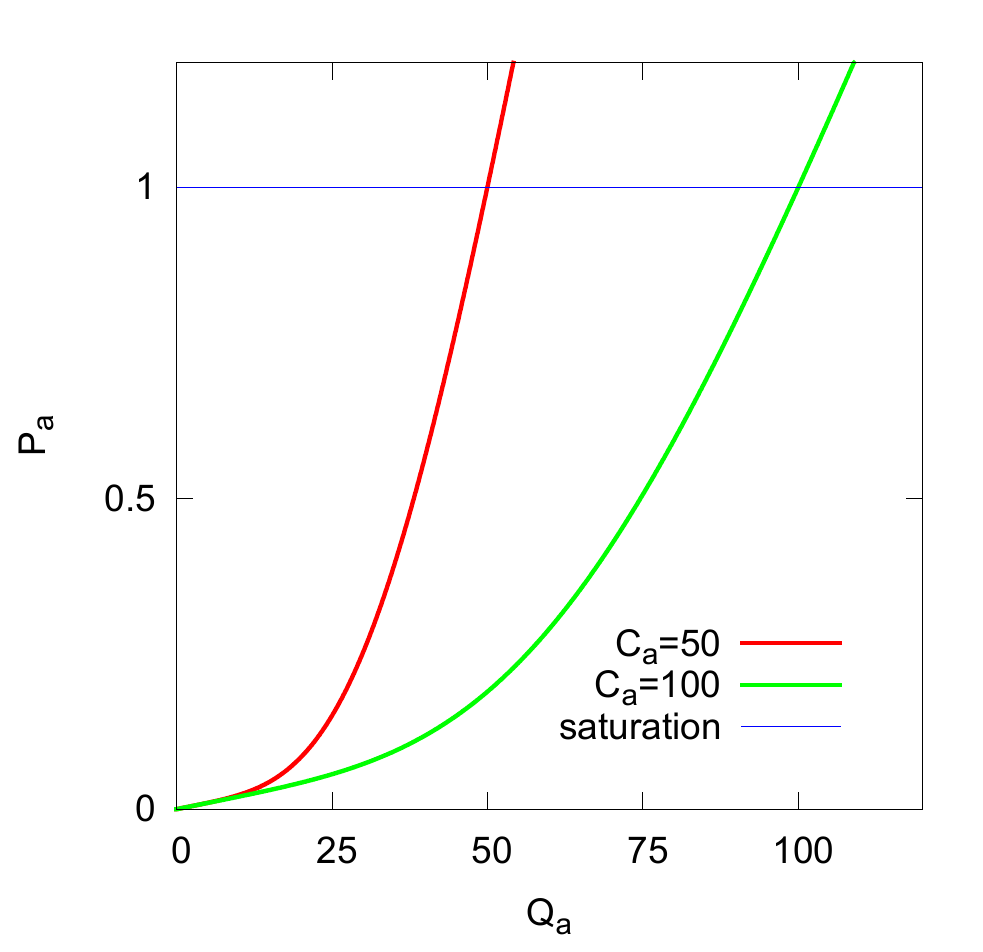}\hfill
\caption{Plot of the convex normalized pressure function with parameters $C_\infty=500$ and $m=4$ for two different congestion thresholds $C_a=50$ and $C_a=100$.}
\label{fig:convex-normalized-pressure}
\end{figure}

\subsection{Work-conservation}

As expected, pressure normalization enables to ensure work-conservation. This can be easily visualized in Figure~\ref{fig:work-conservation} where one can observe that the deadlock of Figure~\ref{fig:deadlock} is resolved. In the following, we prove Theorem~\ref{thm:work-conservation} which states work-conservation under convex normalized pressure.

\begin{theorem}[Work-conservation under normalized pressures]
\label{thm:work-conservation}
Assume that pressure functions $P_a$ are increasing functions taking values in $[0,1]$, $P_a(0)=0$ and $P_a(C_a)=1$. Assume also that in case of equality, the $\arg\max$ of Line~\ref{line-arg-max-bp} in Algorithm~\ref{alg:back-pressure} privileges phases $p_i$ such that there exists $a,b$ with $Q_{ab}(k)>0$, $Q_b(k)<C_b$ and $\mu_{ab}(p_i)>0$. Then, back-pressure control using normalized pressures is work-conserving. 
\end{theorem}
\begin{proof}
Suppose that back-pressure control using normalized pressure functions is not work-conserving. Then, there exists a server at a junction $\junction_i$ which does not work during some slot $k$, while there exists a phase $\tilde p_i$ such that for some $a,b$, $\mu_{ab}(\tilde p_i)>0$, $Q_{ab}(k)>0$ and $Q_b(k) < C_b$.

Let $\pbf_i(k)$ denote the phase at junction $\junction_i$ computed by Algorithm~\ref{alg:back-pressure}. If the server at the junction does not work during slot $k$ under phase  $\pbf_i(k)$, then for all $c,d$ such that $\mu_{cd}(\pbf_i(k))>0$ we have either $Q_{cd}(k)=0$, or $Q_{d}(k) \geq C_d$ (otherwise, the server would work). In the first case, $d_{cd}(k)=0$, and in the second case, $\Pi_{d}(k)=1$. Since $W_{cd}(k) = d_{cd}(k) \max\left(\Pi_{c}(k)-\Pi_{d}(k),0\right)$ and $\Pi_c(k) \leq 1$, we necessarily have $W_{cd}(k)=0$. As a result, for phase $\pbf_i(k)$, we have $\sum_{cd} W_{cd}(k) \mu_{cd}(\pbf_i(k))=0$.
 
On the other hand, by positivity of flows and weights, we have $\sum_{cd} W_{cd} \mu_{cd} (\tilde p_i) \geq 0$. If $\sum_{cd} W_{cd}(k) \mu_{cd} (\tilde p_i) > 0$, it is absurd because $\tilde p_i$ should have been selected by Algorithm~\ref{alg:back-pressure} instead of $\pbf_i(k)$. If $\sum_{cd} W_{cd}(k) \mu_{cd} (\tilde p_i)=0$, it is also absurd because $\sum_{cd} W_{cd}(k)\mu_{cd} (\tilde p_i) = \sum_{cd} W_{cd}(k) \mu_{cd} (\pbf_i(k))$, and again, $\tilde p_i$ should have been selected by Algorithm~\ref{alg:back-pressure} instead of $\pbf_i(k)$ because an equality holds but contrary to $\pbf_i(k)$, there exists for $\tilde p_i$ an input $\node_a$ and an output $\node_b$ with $\mu_{ab}(\tilde p_i)>0$, $Q_{ab}(k)>0$ and $Q_b(k)< C_b$ (see the second assumption in the theorem).
\end{proof}

\begin{figure}[ht]
\centering
\includegraphics[width=1\linewidth]{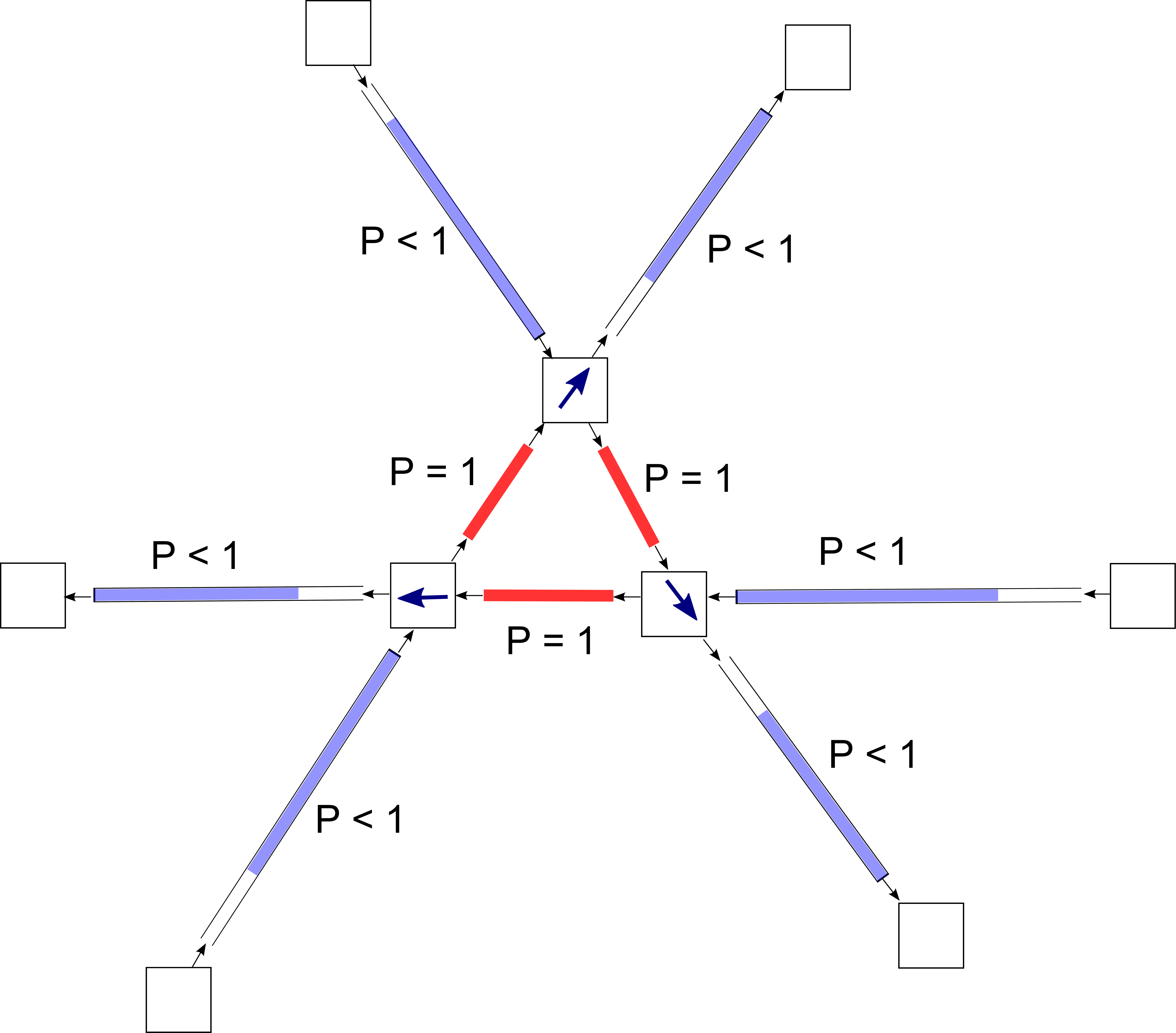}\hfill
\caption{Work-conservation and deadlock resolution using normalized pressure functions. The pressure at full nodes equals 1, while it is strictly lower than 1 at non-full nodes. As a result, the phase computed by Algorithm~\ref{alg:back-pressure} enables to empty full nodes.}
\label{fig:work-conservation}
\end{figure}

\section{Simulation results}
\label{sec:simulations}
In this section, we compare the performance of our proposed capacity-aware back-pressure control with current back-pressure~\cite{Wongpiromsarn2012, Varaiya2009} and a non-optimized fixed cycle traffic light given as reference.
\subsection{Simulation setup}
We have implemented our traffic signal control schemes on the top of the traffic simulator SUMO (Simulation of Urban MObility)~\cite{SUMO2012}. SUMO is a widely recognized open-source traffic simulation package including a traffic simulator as well as supporting tools. The simulator is microscopic, inter- and multi-modal, space-continuous and time-discrete, providing a fair approximation of real world traffic scenarios. 

We adopt a non-uniform network with several types of roads and intersections (Figure~\ref{fig:network}). All roads are bi-directional. Roads V2, V4, V6, V8, H1, H3, H5 and H7 possess only one lane on each direction while the rest of roads have two lanes. Close to the intersection, each road has an additional dedicated left-turn lane. Due to the difference in the number of lanes, there are four types of intersections (Figure~\ref{fig:junction_type}). Each intersection has four phases, as described in Figure~\ref{fig:phases}. The network is open as vehicles may leave and enter the system at all roads.
\begin{figure}[ht]
\centering
\includegraphics[width=1\linewidth]{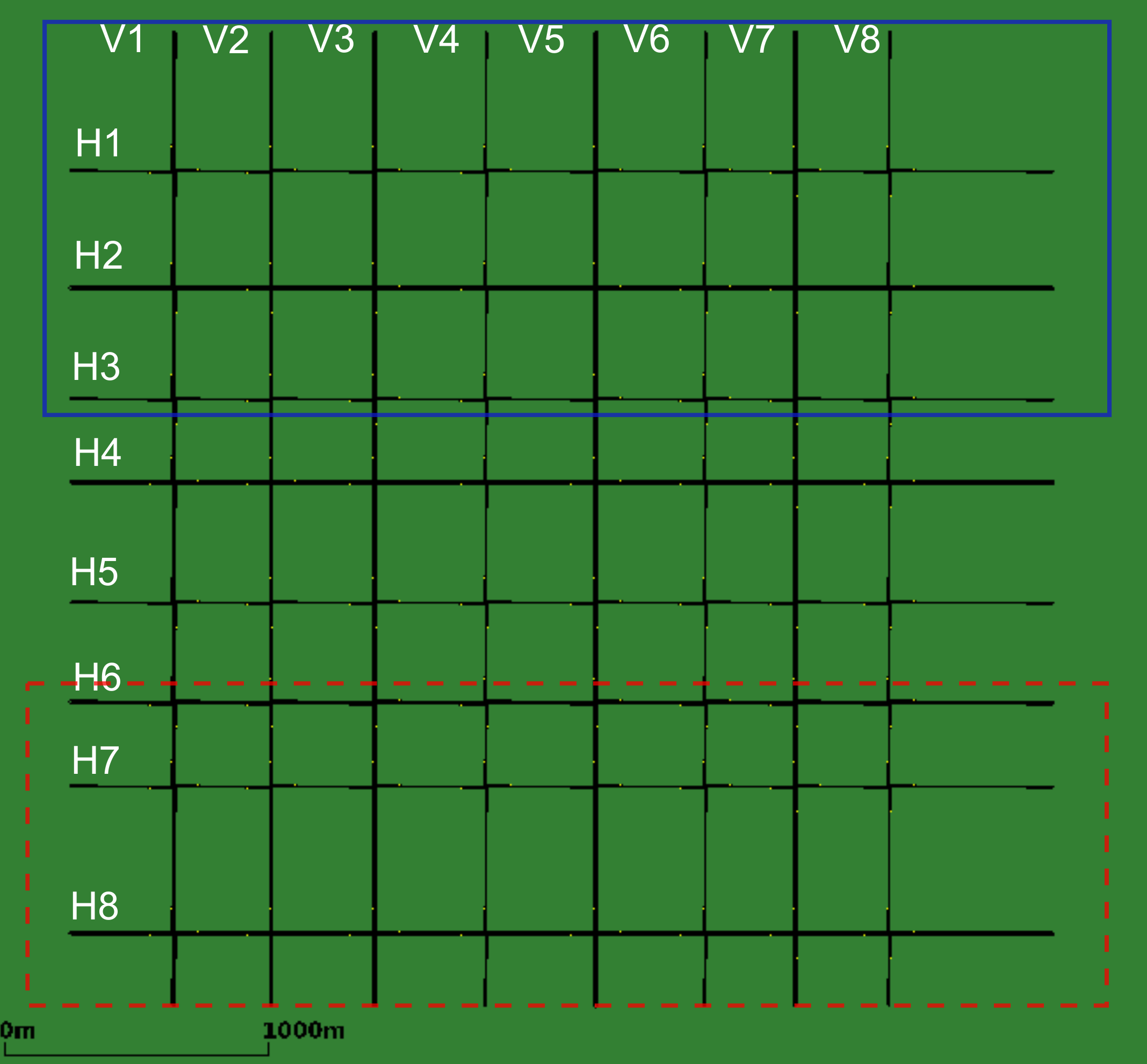}\hfill
\caption{The road network used for simulations.}
\label{fig:network}
\end{figure}

\begin{figure}[ht]
\centering
\includegraphics[width=1\linewidth]{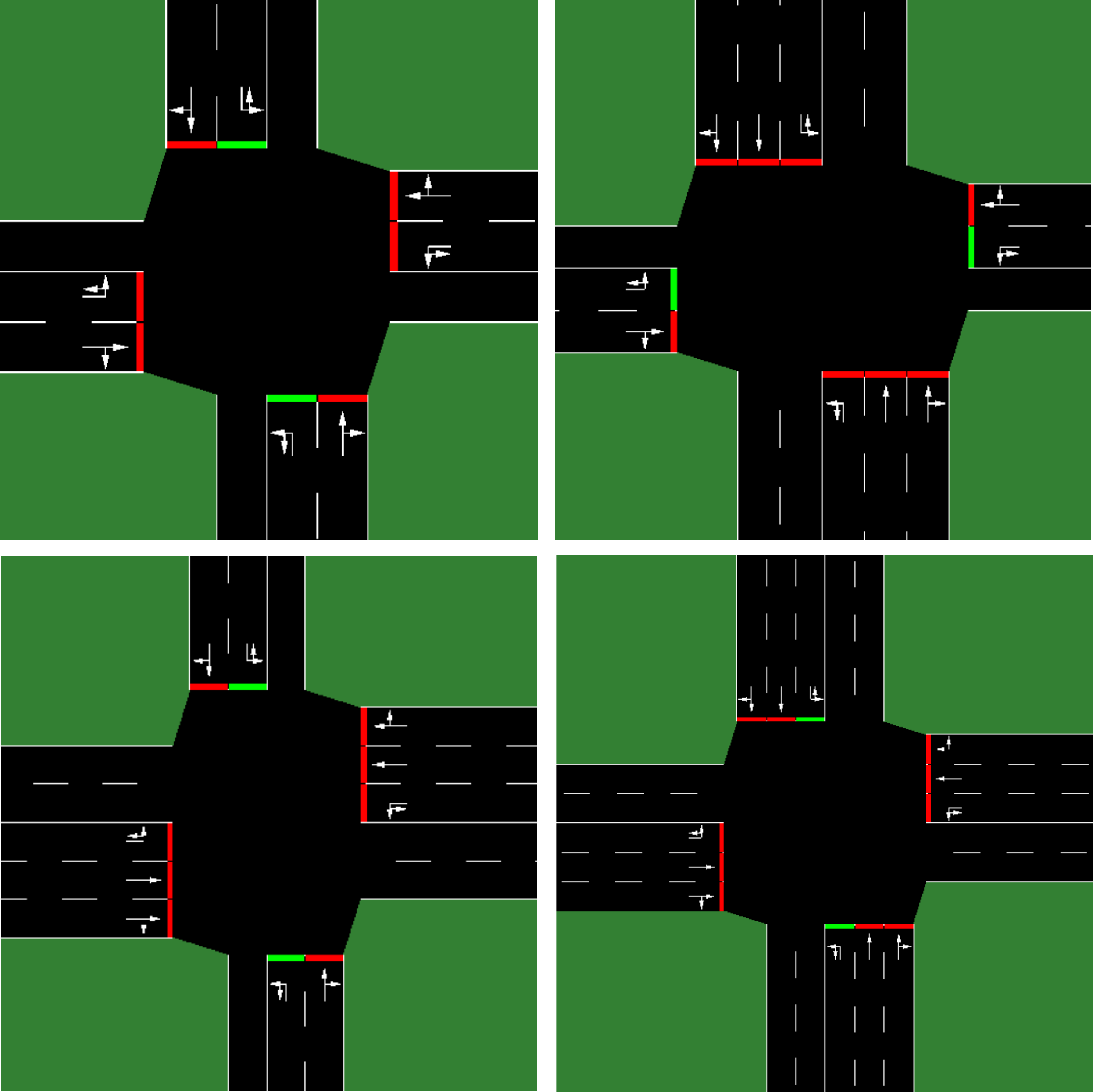}
\caption{Four intersection types depending on the number of lanes of incoming roads.}
\label{fig:junction_type}
\end{figure}

We generate traffic flows using \textit{ActivityGen} available in SUMO supporting tools. \textit{ActivityGen} considers the road network as a city. It takes as inputs, in particular, the city population, the spatial distribution of the population, the working zones, and produces traffic flows with on/off peak patterns. In our network, we set the habitation area to northern city (area in blue rectangle in Figure~\ref{fig:network}) and the working zone resides in the southern city (area in red dashed rectangle in Figure~\ref{fig:network}). We design the test scenarios with a city population ranging from 10000 to 39000 people. We simulate the traffic during a typical 3-hour morning peak (7 am to 9 am). An exemplary histogram of vehicle arrivals in the morning peak is given in Figure~\ref{fig:vehicle_histo}. All Vehicles adopt the default vehicle model of SUMO.
\begin{figure}[ht]
\centering
\includegraphics[width=1\linewidth]{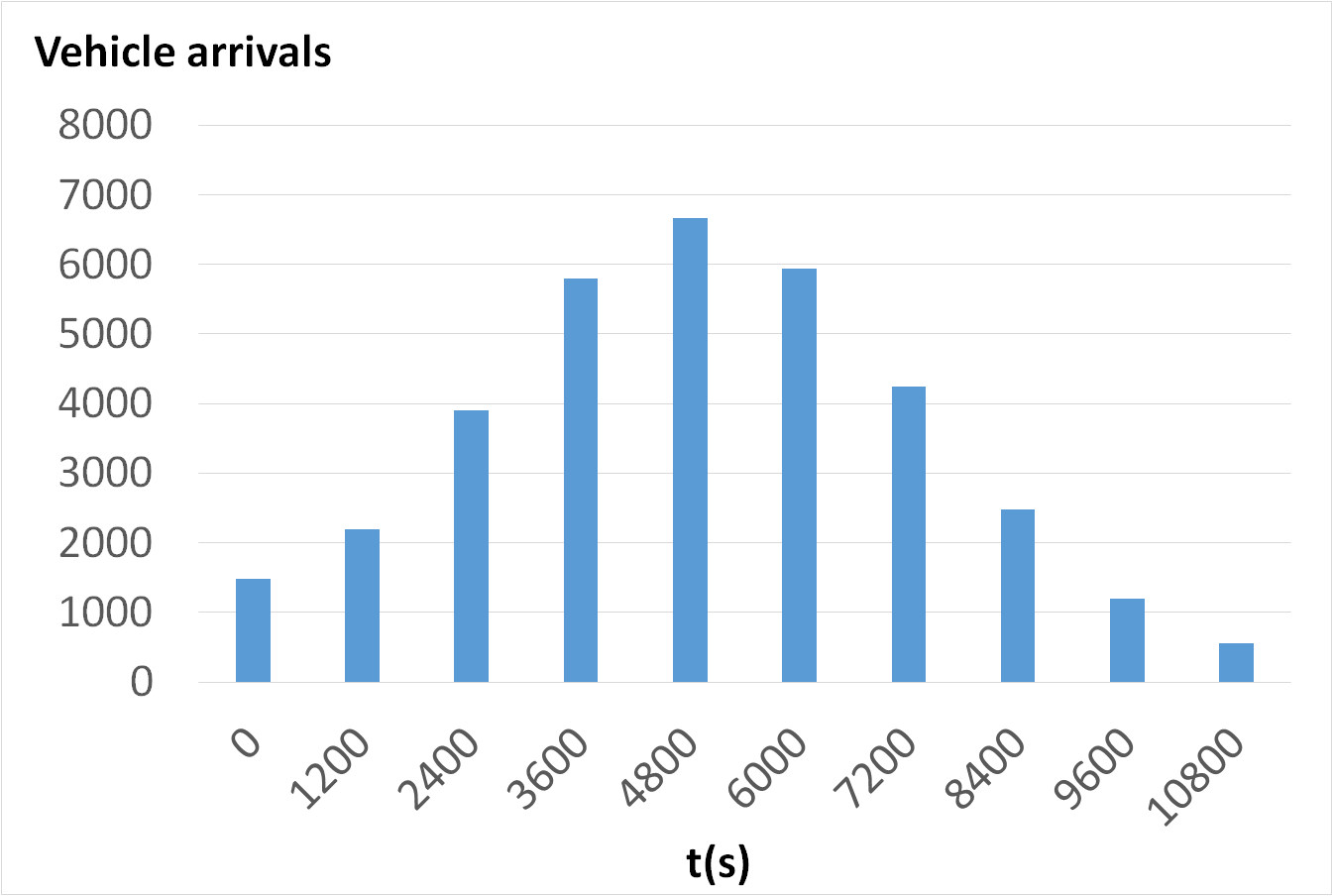}\hfill
\caption{Vehicle arrivals histogram during the morning peak hour (7 am to 9 am) when the city population is set to 33000.}
\label{fig:vehicle_histo}
\end{figure}
A python interface, called TraCI, is provided in SUMO package which enables real-time traffic control. In particular, queues lengths can be retrieved by deploying loop detectors and the phase of traffic lights (more precisely the program) can be updated. We implemented control schemes as a python application. The application interacts with SUMO through TraCI to retrieve queue lengths and to control traffic lights. Three control schemes have been implemented:
\begin{enumerate}
\item Current back-pressure of~\cite{Wongpiromsarn2012, Varaiya2009}: the duration of a time slot is set to 15 seconds, including 4 seconds of yellow phase.
\item Capacity-aware back-pressure: the capacity of roads is computed using the lane length retrieved through TraCI and the vehicle length and minimum gap. The length of the default vehicle in SUMO is 5 meters and the minimum gap is 2.5 meters. As a result, the capacity of a lane is its length in meters divided by 7.5. The time slot setting is the same as in back-pressure: 15 seconds, and the parameters of the convex normalized pressure of Equation~\eqref{eq:pressure-function} are $m=2$ and $C_\infty=200$.
\item Non-optimized fixed cycle control: traffic lights periodically switch the applied phase following the sequence (a)-(c)-(b)-(d) of Figure~\ref{fig:phases}. A complete cycle lasts 60 seconds, in which phase (a) and (b) last 16 seconds and phase (c) and (d) last 6 seconds. The yellow time between two phases is set to 4 seconds. The duration of phases is fixed arbitrarily, it is not optimized.
\end{enumerate}
A complete test run compares the three control schemes under 4 different populations. To ensure comparability, the same random seed is given to \textit{ActivityGen} module for all simulations of a test run.

\subsection{Results and analysis}

Figure~\ref{fig:queue_length} depicts the evolution of the number of vehicles in the network under four population scenarios. We also measure periodically the total time spent by vehicles averaged over vehicles currently in the network, which is a good performance indicator, since users of the network want to minimize their travel time. The evolution of the total time spent is depicted in Figure~\ref{fig:travel_time}. We observe that for a small population of 10000, all control schemes have a similar performance, even if the queue is slightly greater under the non-optimized fixed cycle control scheme. For a population of 27000, back-pressure and capacity-aware back-pressure have similar performance, while they both outperform the fixed-cycle control scheme. As population grows (leading to a growth of flows through the network), for a population of 33000 and 39000, capacity-aware back-pressure outperforms back-pressure control. Therefore, the performance gain of capacity-aware back-pressure mainly occurs at heavy load. Figure~\ref{fig:block} presents a typical configuration that may occur under back-pressure control, while it is alleviated in capacity-aware back-pressure control. Figure~\ref{fig:block}(a) provides a bird eye view on the intersection between H5 and V4. H5 is a 2-lane main street and V4 has only one lane. V4 has reached its full capacity downstream. Figure~\ref{fig:block}(b) offers a closer look on the intersection. Since back-pressure control does not consider the capacity of roads, the left turn phase (phase (c) of Figure~\ref{fig:phases}) of H5 is always activated because the queue difference of corresponding lanes is among the largest. However, it is inefficient and not work-conservative to activate this phase as no vehicle can join the downstream node of V4. Under the same situation, with capacity-aware back-pressure control scheme, the pressure on the downstream node of V4 will be 1, and phase (a) would be activated instead of phase (c) allowing vehicles on H5 going straight to continue their journey (work-conservation is recovered). This qualitative analysis on the example of Figure~\ref{fig:block} is in accordance with quantitative results that show an increase in the average total time spent for back-pressure for a population of 33000 and a drift for a population of 39000 (see Figure~\ref{fig:travel_time}). 

\begin{figure}[ht]
\centering
\includegraphics[width=0.48\linewidth]{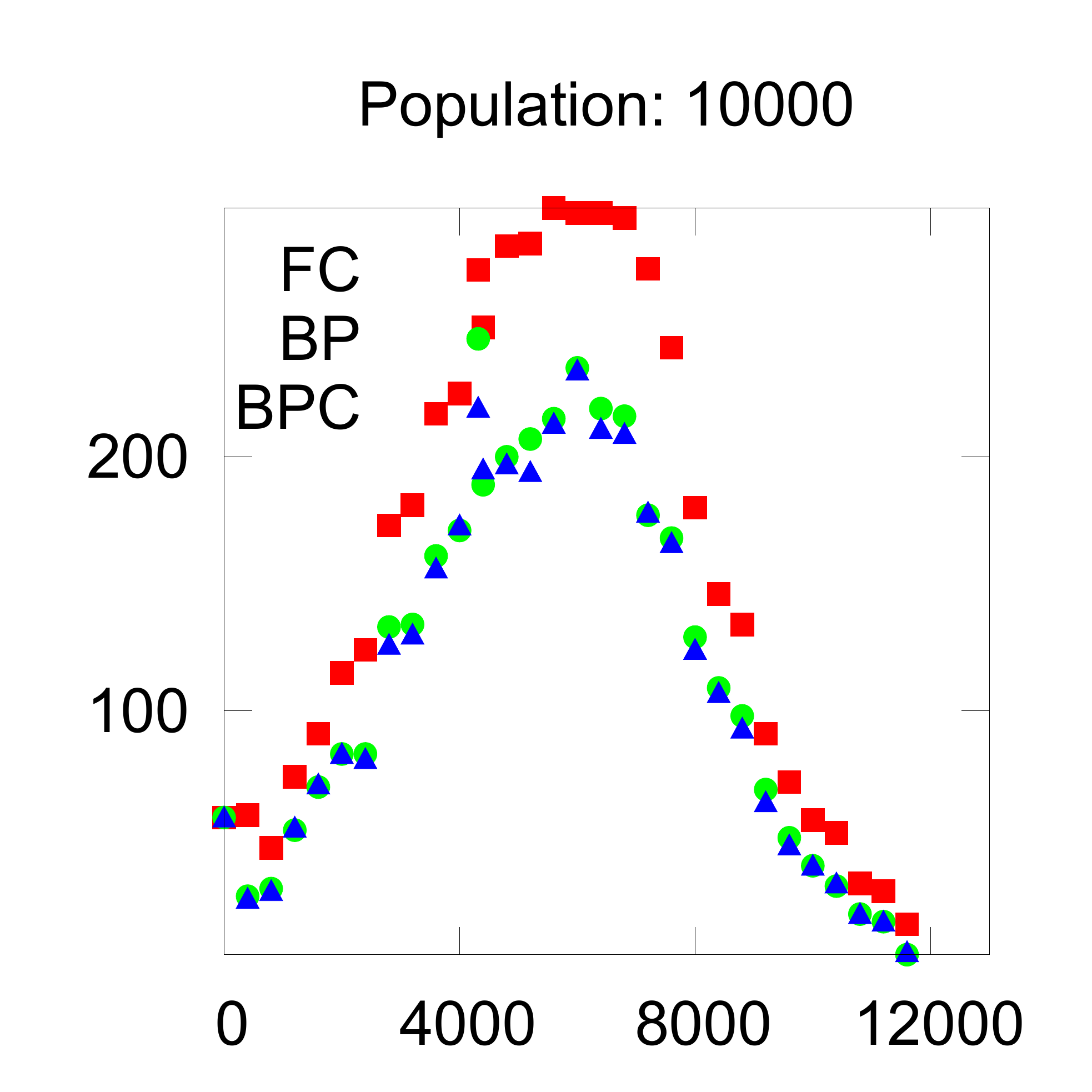}\hfill
\includegraphics[width=0.48\linewidth]{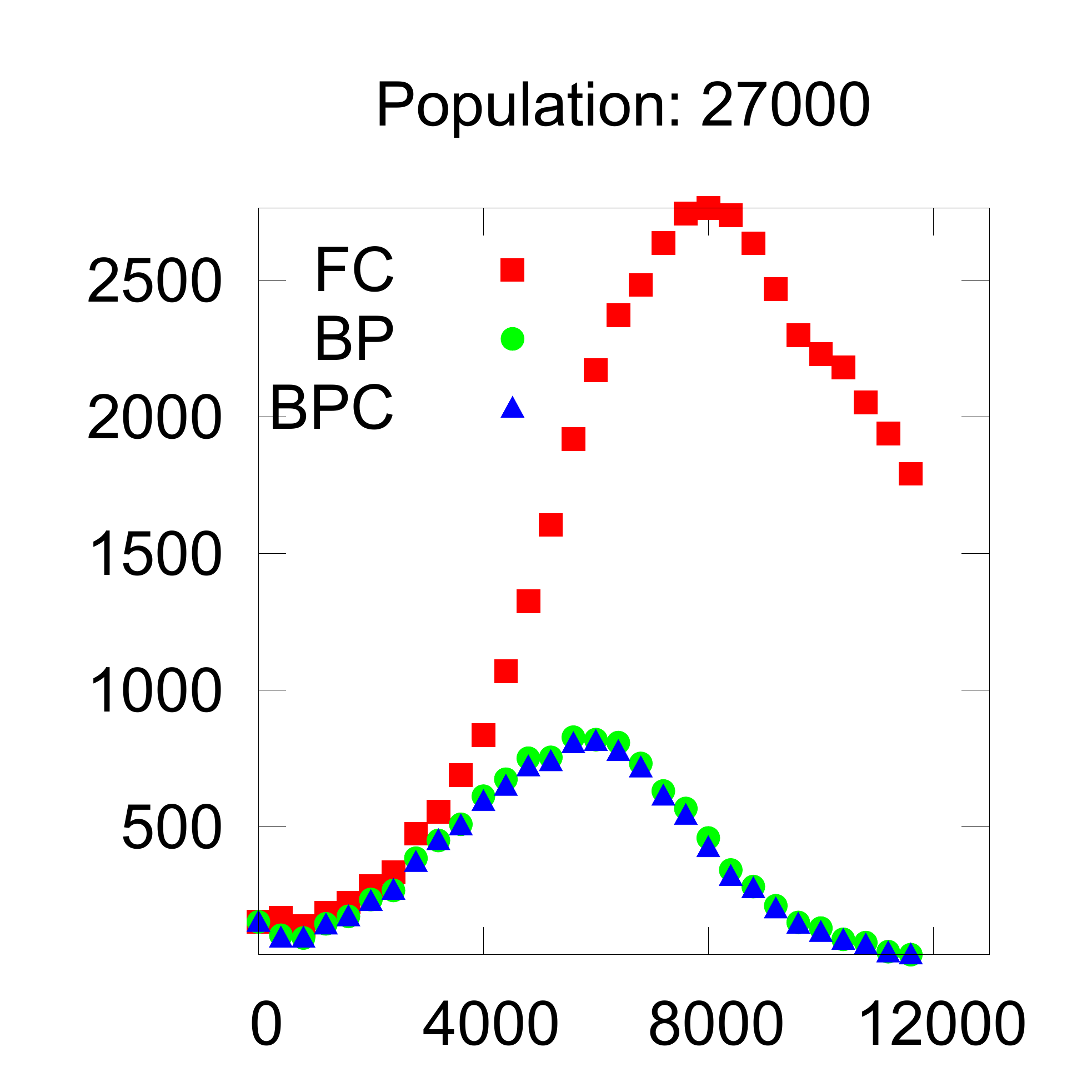}\hfill\\
\vspace{2 mm}
\centering
\includegraphics[width=0.48\linewidth]{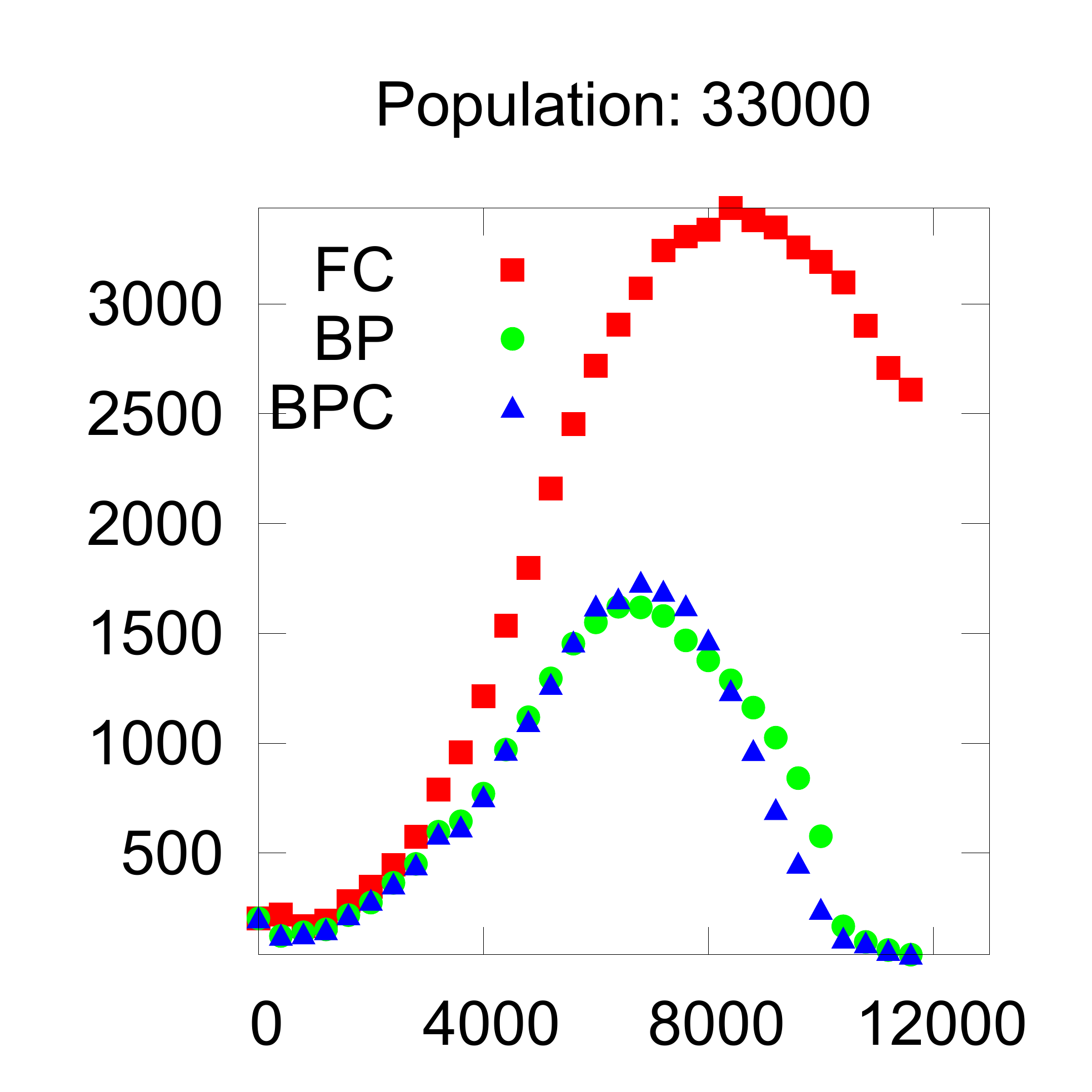}\hfill
\includegraphics[width=0.48\linewidth]{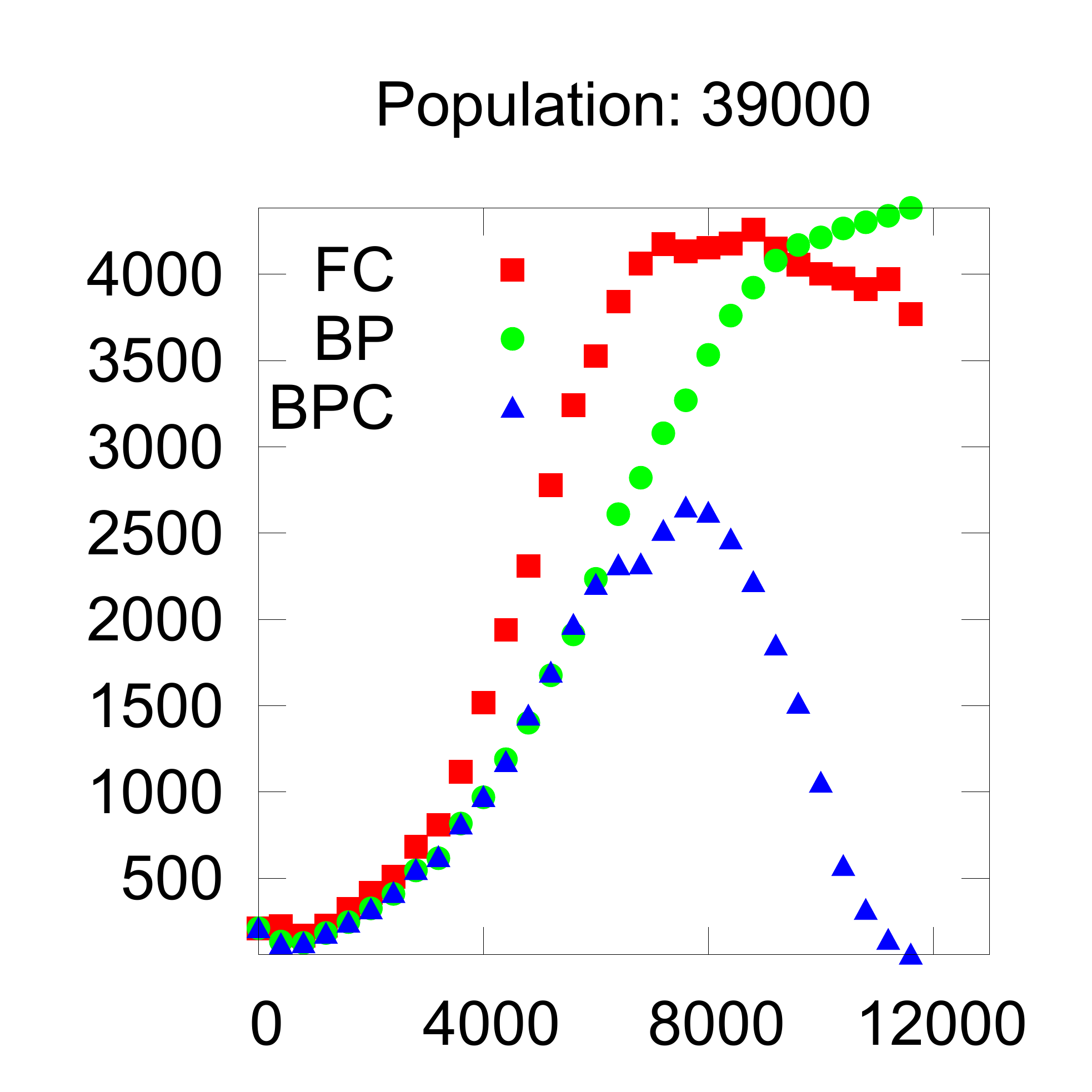}\hfill
\\
\caption{Comparison of the total queue length in the network through time under different population scenarios: 10000, 27000, 33000 and 39000. Time is in seconds. FC refers to the fixed-cycle control scheme, BP to back-pressure and BPC to capacity-aware back-pressure.}
\label{fig:queue_length}
\end{figure}

\begin{figure}[ht]
\centering
\includegraphics[width=0.48\linewidth]{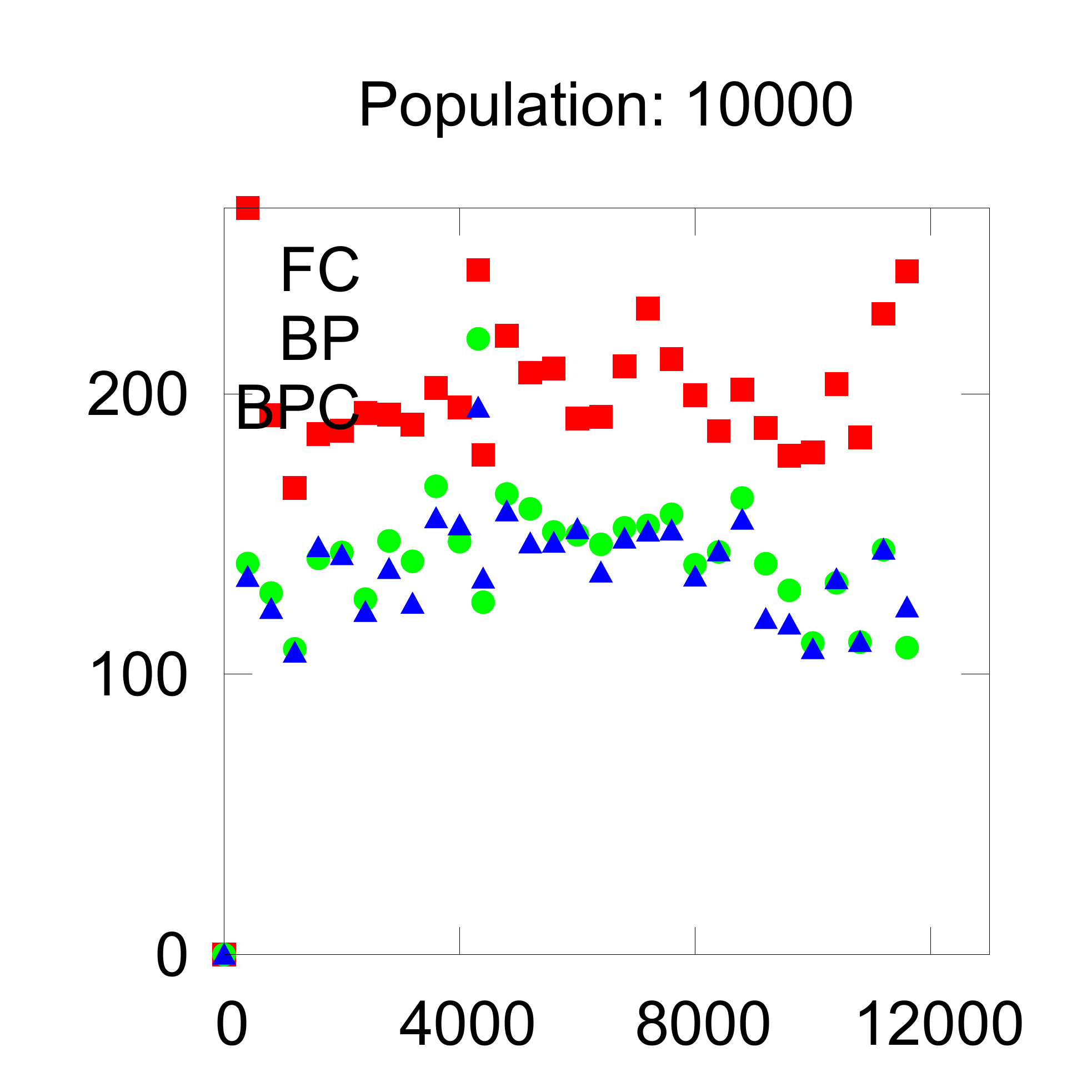}\hfill
\includegraphics[width=0.48\linewidth]{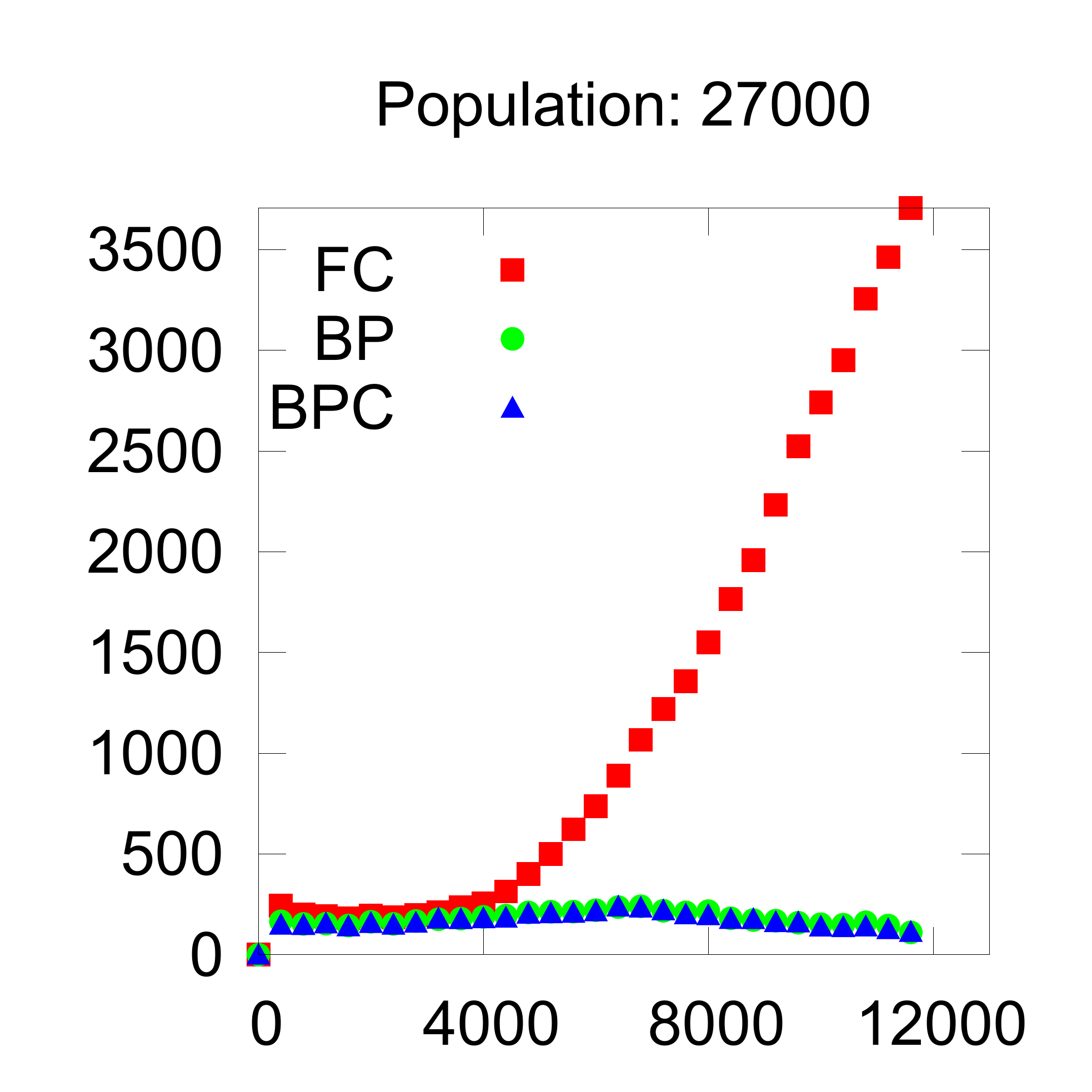}\hfill\\
\vspace{2 mm}
\centering
\includegraphics[width=0.48\linewidth]{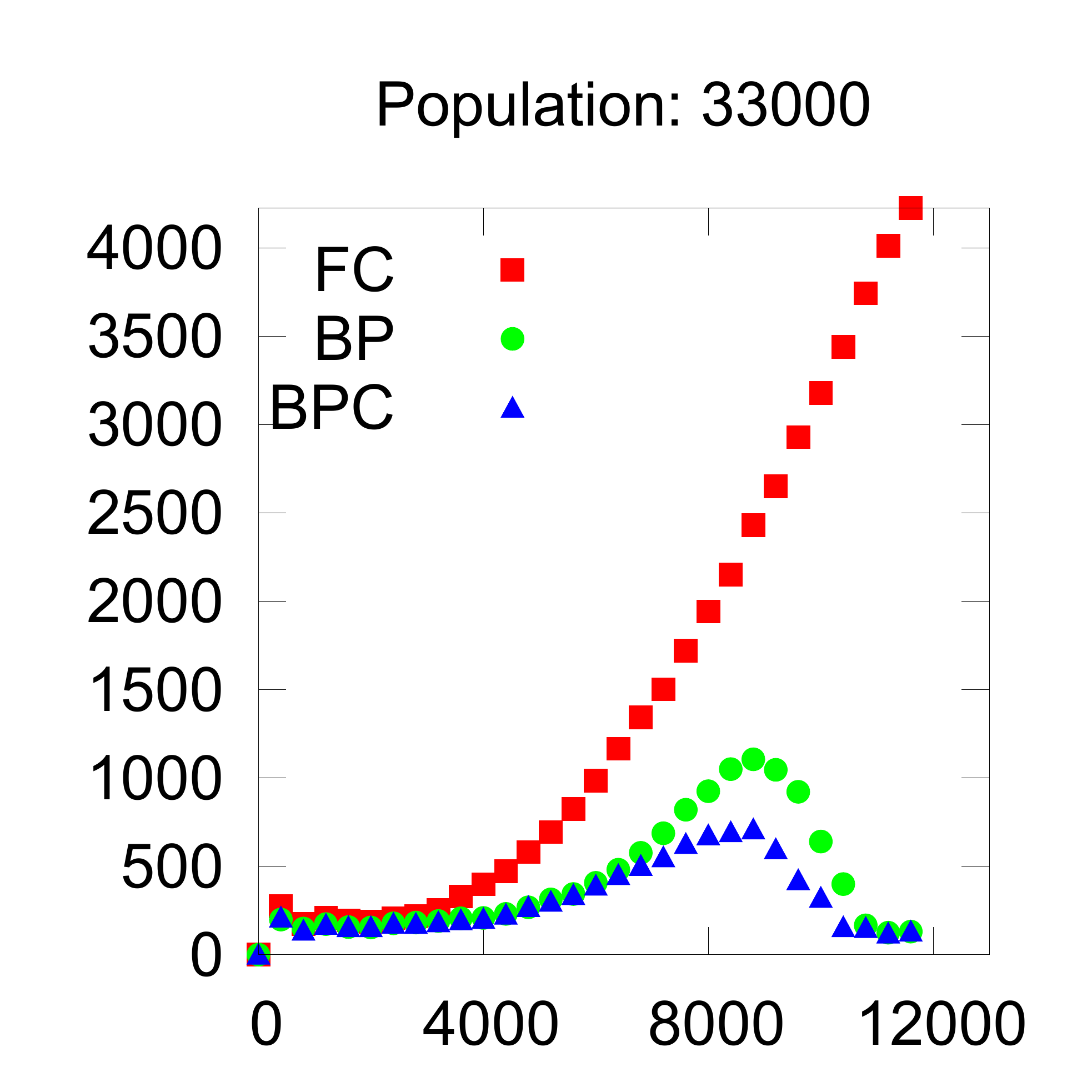}\hfill
\includegraphics[width=0.48\linewidth]{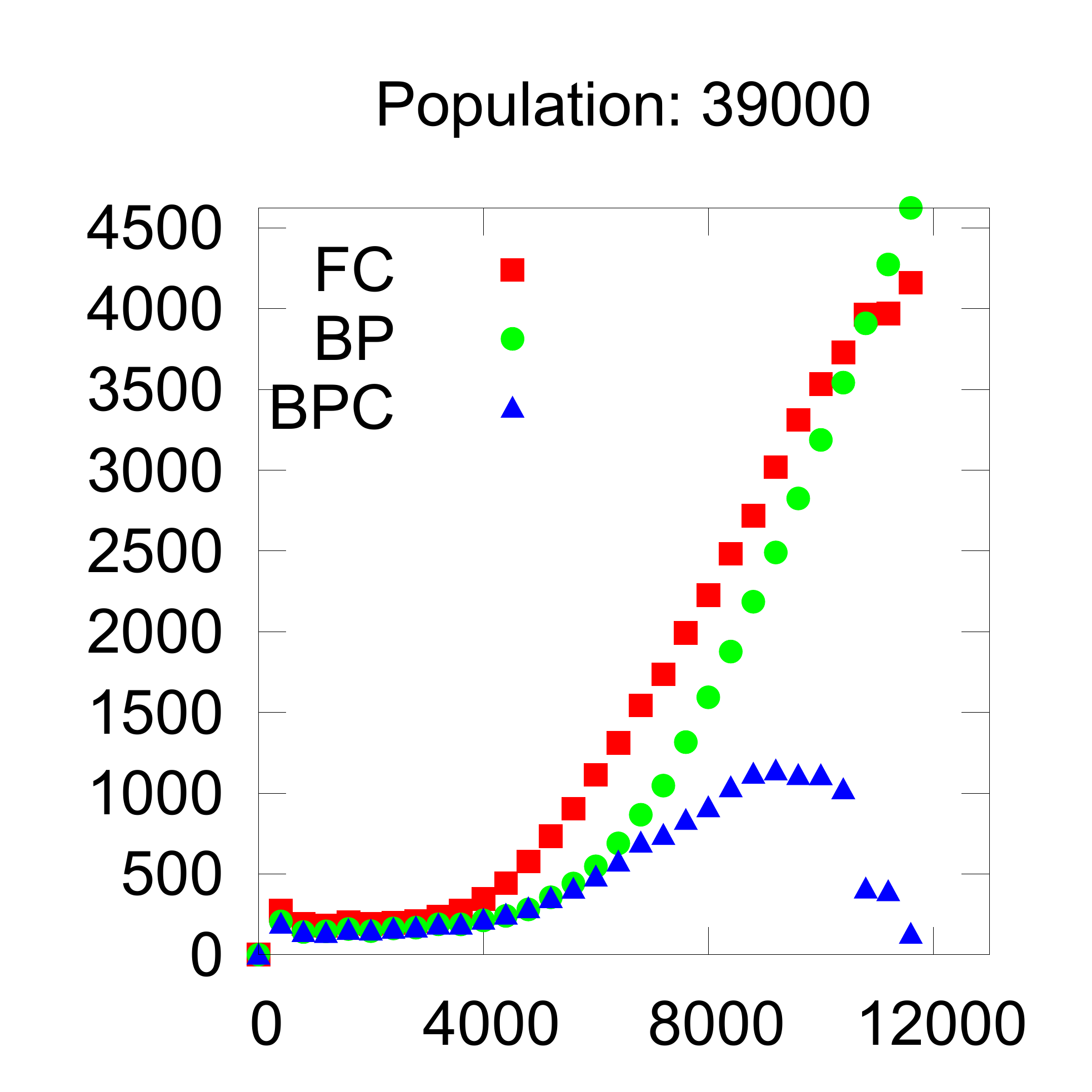}\hfill
\\
\caption{Comparison of the average total time spent in the network through time under different population scenarios: 10000, 27000, 33000 and 39000. Time is in seconds. FC refers to the fixed-cycle control scheme, BP to back-pressure and BPC to capacity-aware back-pressure.}
\label{fig:travel_time}
\end{figure}

\begin{figure}[ht]
\includegraphics[width=1\linewidth]{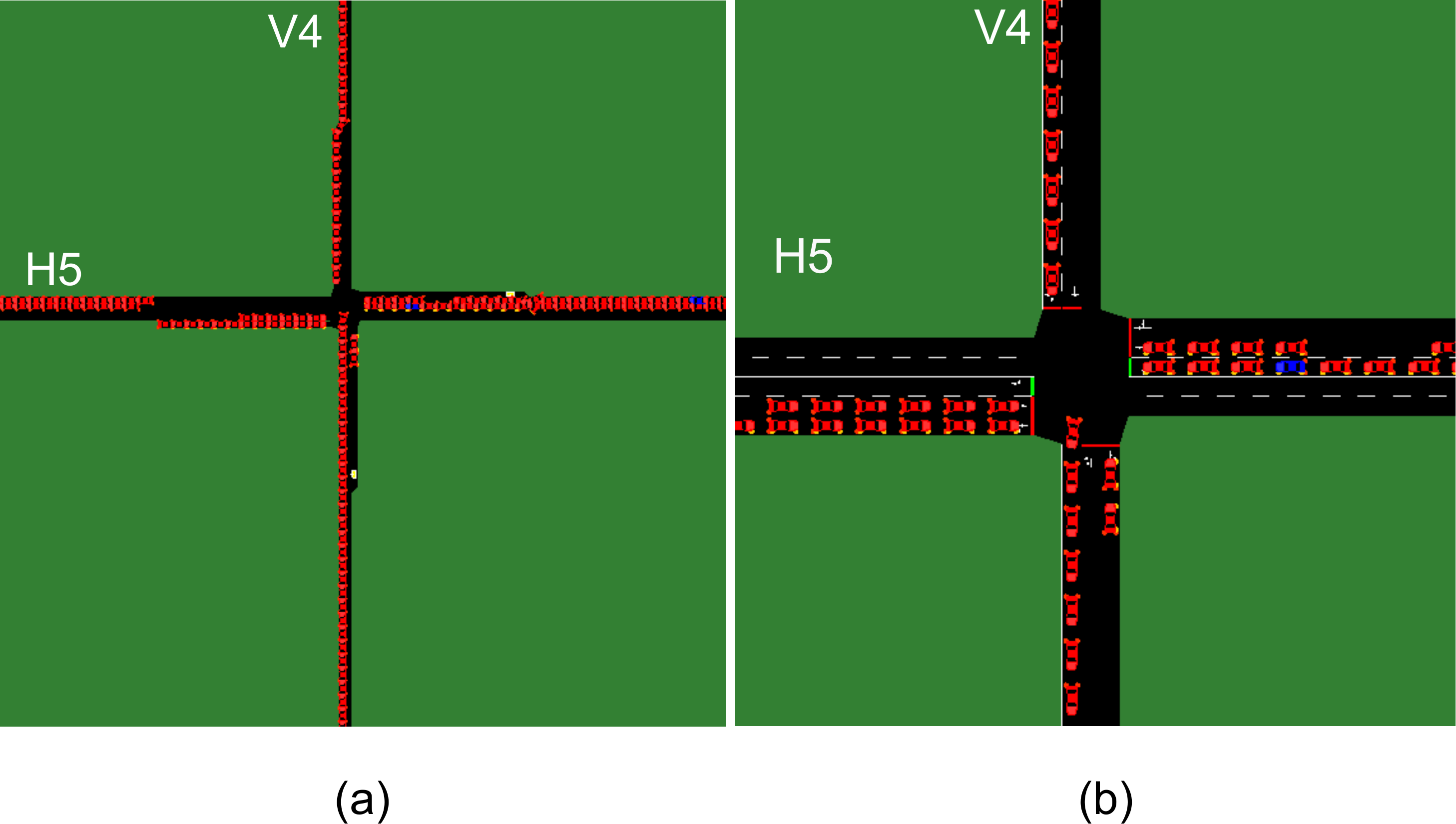}
\caption{A typical inefficient configuration in back-pressure control when capacities are not considered. The screen shot is captured at time 7135 with a population of 39000}\label{fig:block}
\end{figure}

\section{Conclusions and perspectives}
\label{sec:conclusions}

In this paper, we adapt current back-pressure control to take into account bounded queues constraints. The lack of work-conservation of current back-pressure control is proved, and identified as a source of congestion propagation through the network. This phenomenon is caused by pressure saturation at queues that have reached maximum capacity. 

Normalized pressure functions are proved to ensure work-conservation and this property tends to indicate that congestion propagation will be mitigated. Simulations confirm the efficiency of the approach. It is remarkable that performance have been increased under bounded queues constraints as indicated by simulations, while the ability to distribute the control over junctions and $\mathcal{O}(1)$ complexity properties have been conserved. 

However, for very high arrival rates, and in particular above the capacity region, congestion will necessarily eventually propagate through the network. In this case, vast areas of the network will be congested and inter-junctions interactions due to blocking tend to indicate that phase control should be carried out on groups of junctions belonging to the same congested region. Nevertheless, this task is of high complexity due to the exponential complexity of inter-junctions interactions.

Finally, future works on back-pressure signal control should consider the feedback loop between traffic signal control and driver behaviour, and in particular driver routing choice. One can expect drivers at a junction to change their routing choice if the traffic light gives the right-of-way in favour of some particular output nodes due to traffic conditions. It is of high interest to take into account such behaviours, since they may stabilize or unstabilize the queuing network.

\bibliographystyle{ieeetr}
\bibliography{biblio}

\end{document}